\documentclass[journal]{IEEEtran}
\usepackage{amssymb}
\usepackage{amsmath}
\usepackage{graphicx}
\usepackage{epstopdf}
\usepackage{mathtools}
\usepackage{amsthm}
\usepackage{subfigure}
\usepackage{algorithm}
\usepackage{algpseudocode}
\usepackage{enumitem}

\newtheorem{claim}{Claim}

\ifCLASSINFOpdf
\else
\fi
\begin{document}
%
\title{Energy Efficient Barring Factor Enabled Extended Access Barring for IoT Devices in LTE-Advanced}
%
%
%

\author{Prashant K.Wali and~Debabrata Das,~\IEEEmembership{Senior Member,~IEEE}
\thanks{Manuscript received October 26, 2016. This work was supported by a project funded by DeitY, Govt. of India.}
\thanks{Prashant Wali is with International Institute of Information Technology-Bangalore, Bangalore 560100, India (e-mail: iwalihere@gmail.com)}
\thanks{Debabrata Das is with International Institute of Information Technology-Bangalore, Bangalore 560100, India (email: ddas@iiitb.ac.in)}}

\maketitle

\begin{abstract}
Synchronized Random Access Channel (RACH) attempts by Internet of Things (IoT) devices could result in Radio Access Network (RAN) overload in LTE-A. 3GPP adopted  Barring Bitmap Enabled-Extended Access Barring (EAB-BB) mechanism that announces the EAB information (\emph{i.e.,} a list of barred Access Classes) through a barring bitmap as the baseline solution to mitigate the RAN overload. EAB-BB was analyzed for its optimal performance in a recent work. However, there has been no work that analyzes Barring Factor Enabled-Extended Access Barring (EAB-BF), an alternative mechanism that was considered during the standardization process. Due to the modeling complexity involved, not only has it been difficult to analyze EAB-BF, but also, a much more far-reaching issue, like the effect of these schemes on key network performance parameter, like eNodeB energy consumption, has been overlooked. In this regard, for the first time, we develop a novel analytical model for EAB-BF to obtain its performance metrics. Results 
obtained 
from our analysis and simulation are seen to match very well. Furthermore, we also build an eNodeB energy consumption model to serve the IoT RACH requests. We then show that our analytical and energy consumption models can be combined to obtain EAB-BF settings that can minimize eNodeB energy consumption, while simultaneously providing optimal Quality of Service (QoS) performance. Results obtained reveal that the optimal performance of EAB-BF is better than that of EAB-BB. Furthermore, we also show that not only all the three 3GPP-proposed EAB-BF settings considered during standardization provide sub-optimal QoS to devices, but also result in excessive eNodeB energy consumption, thereby acutely penalizing the network. Finally, we provide corrections to these 3GPP-settings that can lead to significant gains in EAB-BF performance.
\end{abstract}

\begin{IEEEkeywords}
Extended Access Barring, Green Communications, Internet of Things, IoT Access, LTE-Advanced,  RAN Overload.
\end{IEEEkeywords}

%
\IEEEpeerreviewmaketitle

\section{Introduction}
\label{intro}
\IEEEPARstart{I}oT devices are expected to drive many useful applications like smart grids and asset tracking \cite{Luigi}. These applications require the devices to report their data periodically to a remote server. Hence, there is a need for an infrastructure to provide them access to the backbone network. LTE-A is seen as a solution for this because of its salient features like high data rates, low latency, support for large bandwidth \cite{3G6} and good coverage. However, since these IoT devices report data periodically with an interval of the order of a few minutes to a day \cite{cheng,cheng2},  they lose uplink synchronization with the eNodeB between each access{\footnote{The device loses uplink synchronization if it does not communicate with the LTE-A base station (called as eNodeB) for 10-20 seconds since the eNodeB UE-INACTIVITY-TIMER expires \cite{timer}.} which can only be regained by completing the RACH procedure. But, for a certain class of devices like the smart meters, a large number 
of them making RACH attempts within a very short period of time results in many preamble collisions. This leads to RAN overload \cite{3G6} severely affecting the QoS (success probability, network access delay) of the devices. To mitigate the RAN overload, several techniques (see for instance \cite{3G6, cheng2, wong,Feng, Ki-Dong} and the references therein) have been proposed among which 3GPP considered Extended Access Barring (EAB) as the baseline solution. Two alternative barring mechanisms were considered for EAB configured devices \cite{cheng3}. 

\begin{itemize}[leftmargin=*]
 \item The first method is based on dividing the devices into 10 Access Classes (ACs) and applying the ON/OFF principle per AC 0-9. With this approach, all devices with EAB configured are either barred or not barred from making RACH attempts depending on their AC. System information update is required to end barring for the prohibited device classes or to change the barred classes. In order to allow all EAB devices to access the network at some point of time, the access opportunities need to be circulated between ACs \cite{cheng3}. In our work, we call this method as \emph{Bitmap Barring Enabled-EAB} and refer to it as EAB-BB.
 \item The second method is based on a probability value (\emph{barring factor}) and a timer (\emph{backoff timer}). In this method, devices that are configured with EAB would generate a random number between 0 and 1 prior to performing RACH. If the random number is lower than the barring factor, the EAB test is passed and devices may attempt RACH. Otherwise, they have to wait a given amount of time indicated by the backoff timer and draw a new random number before attempting RACH again. The \emph{Probability based barring} was considered by proposing three settings (EAB$\left(0.5,16\hspace{2pt}s\right)$, EAB$\left(0.7,8\hspace{2pt}s\right)$ and EAB$\left(0.9,4\hspace{2pt}s\right)$) during the standardization process \cite{cheng, cheng2}. In our work, we call this method as \emph{Barring Factor Enabled-EAB} and refer to it as EAB-BF\footnote{Note that EAB-BF was called EAB \cite{cheng2} before EAB-BB was also considered as an alternative method and adopted by 3GPP with the same name.   Hence we introduce the 
two names EAB-BF and EAB-BB in order to distinguish the two mechanisms from each other.}.
\end{itemize}

EAB-BB was adopted by 3GPP \cite{cheng3} and was analyzed for its optimal performance in a recent work \cite{cheng3}. However, there has been no work that investigates the optimal performance of EAB-BF. Hence, to the best of our knowledge, it is not known whether EAB-BB has the best performance among these two methods that were considered to enable EAB. Moreover in practice, for EAB-BB the system information needs to be updated which is costly both for UEs and the network. Also, the barred/unbarred approach may create peaks to RACH load when the broadcasted access information changes, \emph{i.e.,} when barring is ended, a burst of accesses may occur. In \cite{zhang}, the authors based on their analysis show that long paging cycle limits the performance of the barring phase and short paging cycle limits the performance of the release phase of EAB-BB. Hence, any QoS that requires either a long or a short paging cycle might lead to the aforementioned problems. Also, in \cite{cheng3}
, the authors adopt an algorithm to let the eNodeB to turn EAB on or off according to a congestion coefficient. But it is not clear how the choice of the congestion coefficient threshold value affects transmission periods of SIB14 and the paging cycle and in turn the success probability and mean access delay in general.

Because of the aforementioned reasons, there is a need to investigate and compare the optimal performance of EAB-BB with that of  EAB-BF that can avoid its operational complexities. Additionally, there is also an urgent need to investigate an equally important issue of minimizing the eNodeB energy consumption while ensuring optimal performance. Minimizing the eNodeB energy consumption will have far-reaching benefits because of the following reasons: (\emph{i}) the enormous growth in the cellular industry has already pushed the limits of energy consumption resulting in tens of Mega-Joules of energy being consumed annually, (\emph{ii}) IoT devices have very little data to report in each access, so there is a need to keep the eNodeB energy consumption at minimum while serving them and (\emph{iii}) eNodeBs are the largest energy consumers in a cellular network, taking 60\% of the total share \cite{survey}, so any reduction in their energy consumption will have a substantial bearing on the total network energy 
consumption.

Previous works that investigate methods to reduce eNodeB energy consumption consider only downlink data traffic (see for instance \cite{frenger,strinati,rohit,ps-sps}) in the absence of IoT devices. Works like \cite{Gerasimenko}, which include energy consumption of IoT devices along with their conventional performance metrics, help extend the life of IoT devices but not in scaling down the energy consumption of a cellular network.

Note that the primary purpose of RAN overload control is to relieve the RACH congestion. But it should not stop us from considering an important aspect like saving energy if we can achieve it simultaneously along with RAN overload control. Even though the major power consumption is believed to result from traditional voice/data services, there can be time intervals when this traffic is low to very low. It is known that the traffic in real cellular networks can be modeled as having a periodic sinusoidal profile indicating a large portion of time when the traffic is quite low \cite{bhaskar}. This, added to the fact that these IoT devices continuously make periodic RACH attempts lasting at least 15-20 seconds with a periodicity that could get as low as 5 minutes \cite{3G6} irrespective of other traffic, means that over large fractions of time, the eNodeB spends considerable energy serving these IoT devices. This entails us to investigate if the eNodeB energy consumption can be minimized while effecting RAN 
overload control without compromising on the desired QoS of IoT devices. Hence, this work concentrates only on those intervals when the traditional voice/data traffic is very low or non-existent.

Nevertheless, as a future work, our current analysis can be extended to intervals that have high voice/data traffic by looking at the number of extra subframes required to schedule the downlink traffic after dedicating some for handling the RACH attempts of the IoT devices. Note that even though the energy saving per access cycle in this case could be lesser compared to intervals when the traditional traffic is non existent, the fact that these devices are periodically required to clear RACH, amounts to a huge energy saving in the long run considering around 4 million base stations across the globe \cite{survey}.

Inspite the significant benefits mentioned above, to the best of our knowledge, there have not been any works that have either attempted to find the optimal performance of EAB-BF alone or along with simultaneously minimizing the eNodeB energy consumption.
Although EAB-BF has been analyzed in \cite{cheng2}, it was done only to investigate the performance of the three 3GPP settings that were considered during the standardization process and not to find the settings that ensure its optimal performance. Also, an important element missing in their analysis is the mean access delay.
 
In our work, we derive exact closed expressions for the expected number of devices that attempt and collide in each RACH slot when EAB-BF is employed. We then use these expressions to obtain the success probability and the mean access delay for an attempting device. This kind of an analysis allows us to accurately capture the dynamics in each RACH slot and also mathematically show how the success probability and the mean access delay depend on expected number of devices that attempt and collide, along with EAB-BF parameters. We next build a model for the energy consumed by the eNodeB to serve the IoT devices RACH requests. Combining the aforementioned closed form expressions with the energy consumption model of the eNodeB then helps obtain EAB-BF settings, that can simultaneously minimize eNodeB energy consumption, maximize success probability and minimize mean access delay for IoT devices while satisfying a given QoS constraint. Simulation results for various EAB-BF settings show that the closed form 
expressions obtained through our analysis yield accurate results.  Our work can be termed novel and gains significant importance because of the kind of contributions it allows us to make. Specifically, our contributions are summarized as follows.

\begin{enumerate}[leftmargin=*]
 \item We present a novel analytical model to obtain the success probability and the mean access delay of the IoT devices under the influence of EAB-BF. We also build an energy consumption model of the eNodeB to serve the IoT devices RACH requests.
 
 \item We next illustrate how to combine our analytical model of EAB-BF with eNodeB energy consumption model, to obtain optimal EAB-BF settings, \emph{i.e,} the barring parameters, that minimize the eNodeB energy consumption while simultaneously providing desired QoS to IoT devices during the RACH procedure. Specifically, we illustrate how our analytical expressions can help in building search algorithms to obtain these optimal settings, through an example. To that end, this work could also perhaps motivate newer search algorithms if required.
 
 \item Using the results obtained from the example search algorithm that we present, the optimal performance of EAB-BF is shown to be better than that of EAB-BB which is considered state of the art. This indicates that EAB-BB may not be the best choice to handle densely populated IoT devices in its current form or requires further research to improve its performance.
 
 \item Our results also reveal that all the three 3GPP-proposed settings for EAB-BF considered during standardization not only provide sub-optimal QoS to devices, but also result in excessive eNodeB energy consumption, thereby acutely penalizing the network. Our study then shows that corrections can be made to these 3GPP settings so that the energy consumed by the eNodeB and the mean access delay can be reduced by 50\% with more than 50\% gain in success probability compared to one of the three 3GPP settings. A 50\% reduction in energy consumption along with a substantial reduction in access delay can be obtained against the other two 3GPP settings without any degradation in the success probability. 
 
 This is in contrast to the general belief that a gain in success probability can be achieved only at the cost of mean access delay or vice-versa. In this regard, a key insight that our work provides is that a decrease in mean access delay is possible either with an increase or without any decrease in success probability if we move from a sub-optimal performance point to an optimal performance point or in other words, if we modify a sub-optimal setting to make it an optimal setting.  
 
\end{enumerate}

The rest of the paper is organized as follows. The system model is developed in Section \ref{system-model-section}. In Section \ref{energy-consumption-section}, we develop the LTE-A eNodeB energy consumption model to serve IoT devices RACH requests. In Section \ref{analysis-section}, we define the performance metrics which motivates the analysis of EAB-BF. The joint optimization of EAB-BF and eNodeB energy consumption is investigated in Section \ref{optimization-section} in which we also present an example search algorithm that can be built from our analytical expressions to obtain the optimal performance of EAB-BF. Results obtained from this algorithm are then used to compare with the performance of EAB-BB and the 3GPP proposed EAB-BF settings. Section \ref{conclusion-section} concludes the paper.

\begin{table}
\footnotesize
\centering
\footnotesize \caption{SUMMARY OF THE SYMBOLS USED}
\begin{tabular}{ ||p{0.7cm}|p{4.9cm}|p{1.9cm}||  }
 \hline
 \textbf{Symbol} & \textbf{Parameter} & \textbf{Typical Value used in simulations} \\
 \hline
 $N$ & Total number of IoT devices in the cell & 30000
 \\
 \hline
 $K$  & Number of preambles & 54
\\
\hline
$W$ & Backoff Window interval after collision in msec & 20 ms
\\
\hline
$r_p$ & Rach periodicity in msec & 5 ms
\\
\hline
$P_{eab}$ & EAB Barring Probability & 0.5, 0.7 and 0.9
\\
\hline
$T_{eab}$ & EAB Backoff Time & 16 s, 8 s and 4 s
\\
\hline 
\end{tabular}
\label{table:1}
\end{table}
\section{System Model}\label{system-model-section}
We consider a set of IoT devices already registered with the eNodeB. The devices get activated periodically and try to access the network to send data. The periodicity can range from a few minutes to a few hours or days \cite{3G6}. We assume that each access period (called as \emph{cycle} henceforth) begins at time $t=0$ for simplicity. Since a device is inactive for at least a few minutes between each network 
access, it gets detached from its serving eNodeB. Hence, each device tries to re-establish uplink synchronization with the eNodeB through the RACH procedure in every cycle, after it gets activated at time $0 \leq t \leq T_A$ according to  a density function $p(t)$. The density function $p(t)$ could be derived from a truncated beta distribution specified by 3GPP \cite{3G6} as follows,
 
     \begin{equation}\label{pt}
      p(t) = \frac{t^{\alpha-1}(T_A-t)^{\beta-1}}{T_A^{\alpha+\beta-1}Beta(\alpha,\beta)},
      \end{equation} where, $\alpha, \beta > 0$ and $Beta(\alpha,\beta)$ is the Beta function. 
      
      Denoting the total number of IoT devices in the cell by $N$, the access intensity $\lambda_i$ of the IoT devices is given by \cite{3G6},
      \begin{equation}\label{access-intensity-eqn}
      \lambda_i = N\int_{t_{i-1}}^{t_i}p(t)\mathrm{d}t,
      \end{equation} where, $t_i$ is the time of the $i^{th}$ RACH slot within that access cycle. We assume that within the activation time $T_A$, there are $M$ RACH slots as shown in Fig. \ref{rach-cycle}. The RACH  periodicity which is the length of the interval between two successive RACH slots is denoted by $r_p$. Let the interval $\left[t_{i-1}, t_{i-1}+r_p\right]$ = $\left[t_{i-1}, t_{i}\right]$ be called interval $i$. All the devices that get activated in the interval $i$ are considered as new arrivals for the $i^{th}$ RACH slot. Note that all the devices get activated within time $T_A$, and for $t > T_A$, devices that experience collisions can re-attempt the RACH procedure.
\begin{figure}
\centering
\includegraphics[width = 3.5in]{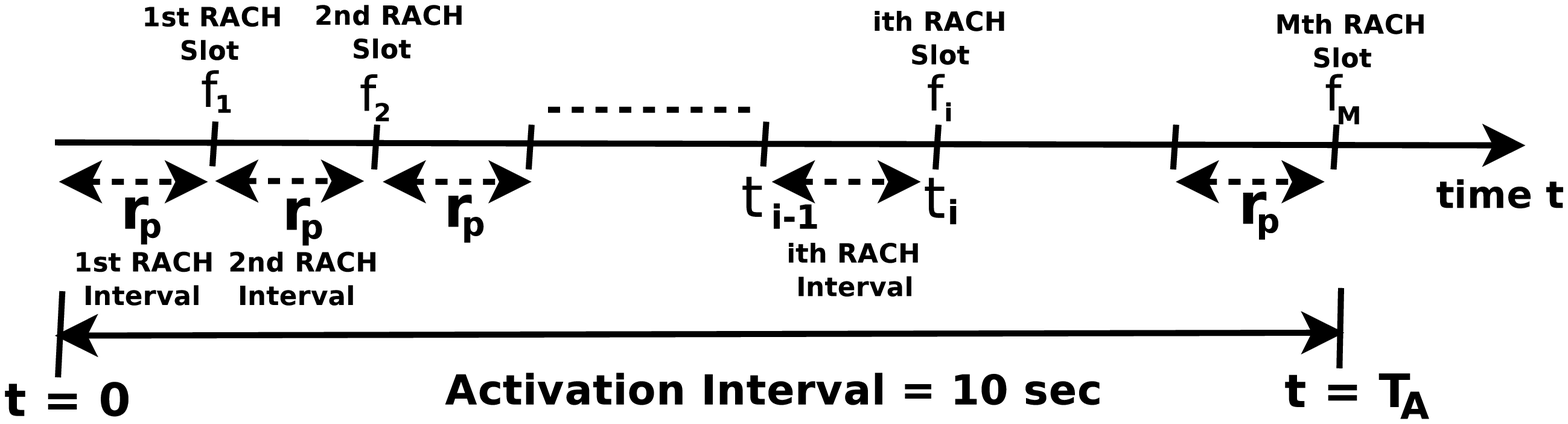}
\caption{System Model for RACH Slots and New Arrivals}
\label{rach-cycle}
\end{figure}

EAB-BF is denoted by EAB$\left(P_{eab},T_{eab}\right)$.  For clarity, we use ``EAB-BF'' to indicate the EAB mechanism that employs the barring factor method and ``EAB test'' to indicate the barring test that a device takes by generating a random number, to be allowed to make a RACH attempt. As explained in the previous section, whenever a device fails in the EAB test (by generating a number larger than $P_{eab}$), it backs off for a  fixed time $T_{eab}$ seconds and repeats this procedure until a number smaller than $P_{eab}$ is chosen \cite{cheng2}\footnote{Note that in another mechanism called as Access Class Barring (ACB) also, the device generates a random number between 0 and 1 and backs off if the generated number is greater than the barring probability, but the back off time is random \cite{3gpp-acb} and hence is different from EAB-BF. Moreover, since our goal in this work is to analyze the optimal performance of EAB-BF \cite{cheng2} and compare it with EAB-BB \cite{cheng3} which was adopted by 3GPP 
over EAB-BF, we consider a fixed back off time as suggested in \cite{cheng2}. To the best of our knowledge, this is the first work which derives the optimal performance of EAB-BF and compares with that of EAB-BB.}. If the device passes the EAB test (by generating a number less than $P_{eab}$), it is allowed to start the RACH procedure by sending a preamble to the eNodeB after choosing it randomly from a set of 54 sequences. A collision occurs when two or more devices choose same preamble sequences resulting in RACH failure{\footnote{This is a conservative assumption as a UE with much stronger signal than others may be selected even in the 
event of collision due to the \emph{capture effect}. But we ignore this possibility in our work.}. If the preamble collides, the device backs off for a time that is uniformly distributed over $\left[0,W-1\right]$ ms and again performs EAB test after returning from backoff. It follows this 
procedure after each collision. On the other hand, if the preamble does not collide, then, we assume that the preamble is detected by the eNodeB\footnote{There is a chance of the preamble not being detected even when there is no collision, albiet, with a very low probability in LTE-A since the random access preambles in LTE-A are normally 
orthogonal to other uplink transmissions \cite[Section 19.3.1.2]{stefania}. Hence we ignore this event.}, in which case, the eNodeB acknowledges  with a Random Access Response (RAR) message within the RAR window (RAR window is a set of subframes in the downlink in which the eNodeB sends the RAR messages). The RAR message is used to send the timing advance information, uplink resource grant for message 3 to be sent by the UE and temporary identify the UE among other things. After some RAR processing time, the device transmits Radio Resource Control (RRC) connection request message via the Physical Uplink Shared Channel (PUSCH) using the resources granted by RAR. As eNodeB needs to establish which device sent which preamble, collision resolution process is required. The RACH procedure ends with a successful reception of RRC connection set-up message  from eNodeB. We assume that both $T_{eab}$ and $W$ are integer multiples of the RACH periodicity $r_p$ for simplicity.

Depending on the QoS requirement of the devices, the eNodeB has to decide the optimal settings of EAB-BF, \emph{i.e.,} $P_{eab}^{*}$ and $T_{eab}^{*}$, that provide the desired QoS with minimum eNodeB energy consumption. To make this decision, the eNodeB holds a database (which can be created offline by the network operator and uploaded into eNodeB as indicated even in \cite{cheng3} or can be created by the eNodeB itself) of the minimum energy consuming settings for each of the possible QoS requirements of the IoT devices. The eNodeB can know about the QoS requirement of devices during registration through the establishment cause information in RRC Connection Request Message sent by the devices during the initial RACH procedure \cite{TS-331}. The eNodeB then has to communicate optimal barring parameters pair $\left(P_{eab}^{*},T_{eab}^{*}\right)$ once to devices, unlike EAB-BB, in which the system information needs to be updated even within an access cycle when the devices are attempting RACH which is costly 
both for UEs and the eNodeB.

Even though the optimal barring parameters pair $\left(P_{eab}^{*},T_{eab}^{*}\right)$ have to be recomputed whenever the number of IoT devices $N$ changes in the cell, note that in the current context, the devices are either expected to be deployed by the network operator or their count is assumed to be known and remain unchanged \cite{3G6}. Similar assumption is also made in \cite{cheng3}, wherein, the authors provide optimal values of transmission periods for SIB14 messages and paging cycle for a set of $N$ devices. The device count being static is justified by the fact that the number of devices do not change frequently, and might also be regulated by the operator. In such a scenario, then, we can assume that the new optimal settings could be obtained by the presented algorithm before the next access cycle of the devices begins, since its periodicity is at least few minutes. Alternatively, the operator could bar the new devices from accessing the network till a new set of optimal settings are obtained 
for the updated $N$. In either case, the maximum count for the devices within a single LTE-A cell is currently specified to be 30000, as per the 3GPP specifications \cite{3G6} and hence the time taken for the search can be well within the acceptable bounds. In view of the above arguments, the algorithm presented in this work could be either used to create the database offline and fed into the eNodeB or could run in the eNodeB itself without proving to be too costly. Also, EAB-BF does not encounter the problem of some QoS requirements limiting its performance unlike EAB-BB, in which if those QoS requirements demanded a long or short paging cycle could limit the performance of the barring phase and release phase.

Overall, the main goal and significance of our work lies in building a novel analytical framework to obtain the performance metrics of EAB-BF and then exhibiting a way to build search algorithms using these performance metrics, to obtain its optimal settings, which has been illustrated through an example algorithm. Additionally, the algorithm presented in this paper also explains how the near-optimal performance values and settings of EAB-BF were obtained in our work. In this regard, this work could perhaps motivate newer search algorithms if required. Equally important is its contribution in being able to help us compare EAB-BF with EAB-BB, and showcase its superior performance.

\section{ENodeB Energy Consumption to serve IoT RACH Requests}\label{energy-consumption-section}
Note that in LTE-A, if a preamble is successfully transmitted, the actual data will then be transmitted without contention on Physical Uplink Shared Channel (PUSCH) via scheduled transmissions and the time it takes is fixed \cite{wong}. Therefore, the time taken (or mean number of attempts required) for all the IoT devices to successfully transmit Step 1 preamble sequences during RACH is the dominant part contributing to the QoS provided to IoT devices. Hence, we concentrate on the performance only during the first step of the RACH procedure. Therefore, RAR messages that use PDSCH (Physical Downlink Shared Channel) resources in the downlink subframes and transmitted in response to preamble transmissions by the devices are identified as energy consuming events. 

In this section, we will first provide a brief background of LTE-A downlink physical layer structure. We will then point to a set of 3GPP specifications to derive an exact expression for maximum number of RAR messages that can be sent in each subframe of the RAR window. Next, we use this model to evaluate the energy consumed by the eNodeB when the IoT devices make RACH attempts in a cycle.  The notations used and their meaning are summarized in Table \ref{table-for-RAR}.
\begin{table}
\footnotesize
\centering
\footnotesize \caption{SUMMARY OF THE SYMBOLS USED. TYPICAL VALUES USED IF ANY ARE SHOWN IN BRACKETS.}
\begin{tabular}{ ||p{0.8cm}|p{7.2cm}||}
 \hline
 \textbf{Symbol} & \textbf{Parameters represented and Typical Values used}
 \\
 \hline
 $B$ & Bandwidth of the LTE-A cell (\textbf{5 MHz})
 \\
 \hline
 $M_{cs}$ & Modulation and Coding Scheme (MCS) used to transmit RAR messages (\textbf{0.930})
 \\
 \hline
 $L_b$  & Number of physical bits  required to transmit one RAR message (\textbf{61})
\\
\hline
$D_s$ & Number of symbols in a PRB that can carry data bits (\textbf{120})
\\
\hline
$N_{rar}^{prb}$ & Number of RAR messages that can be included in one PRB (\textbf{4})
\\
\hline
$N_{prb}$ & Number of PRBs available in a subframe (\textbf{25})
\\
\hline
$N_{max}^{rar}$ & Maximum number of RAR messages that can be sent in one subframe (\textbf{39})
\\
\hline
$N_{sf}$ & Number of subframes required to send $r$ RAR messages
\\
\hline
$W_{rar}$ & Size of the RAR window (\textbf{5 subframes})
\\
\hline
$P_o$ & Fixed power consumption part of eNodeB (\textbf{170 W})
\\
\hline 
$P_{t}$ & Power transmitted per PRB (\textbf{0.8 W})
\\
\hline
$\delta$ & Power Amplifier Efficiency (\textbf{30\%})
\\
\hline
$E_{c}^{k}$ & Energy consumed to transmit data on $k$ PRBs
\\
\hline
$E_{in}^{r}$ & Energy consumed to send $r$ RAR messages
\\
\hline
\end{tabular}
\label{table-for-RAR}
\end{table}
\subsection{PRB Requirement for RAR Messages in LTE-A}
In LTE-A, each downlink frame is of 10 ms duration and consists of ten subframes. Each subframe is of 1 ms duration. A subframe consists of two 0.5 ms slots, with each slot containing seven OFDM symbols. In the frequency domain, the system bandwidth $B$ is divided into several subcarriers, each with a bandwidth of 15 kHz. A set of 12 consecutive subcarriers for a duration of 1 slot is called a Physical Resource Block (PRB). The total amount of available PRBs during one subframe depends on the number of subcarriers, \emph{i.e.}, the total bandwidth allocated for the LTE-A carrier. For example, 10 MHz cell bandwidth would correspond to 600 subcarriers and 100 PRBs in each subframe. In LTE-A, data is transmitted over a pair of PRBs (in time domain) called as Transmission Time Interval (TTI).

Although all the PRBs available in the total bandwidth can be allocated for sending RAR messages in each subframe within the RAR window, the following constraints are imposed for the transmission of RAR messages:
\begin{enumerate}[leftmargin=*]
 \item Only QPSK modulation is allowed for RAR messages and the effective channel code rate determines the number of information bits being transmitted \cite{3G9}. There are tables \cite[Tables 7.1.7.1-1, 7.1.7.2.1-1 and 7.1.7.2.3]{3G9} that map the effective channel code rate to the number of RAR information bits being transmitted.
 \item Irrespective of the bandwidth, a maximum of 2216 RAR information bits can be transmitted in a single subframe within the RAR window \cite{3G10}.
 \item The effective channel code rate chosen should ensure that the efficiency, \emph{i.e.,} the RAR information bits per symbol $\leq$ 0.930 \cite{3G9}.
\end{enumerate}

Each RAR message requires 7 bytes, \emph{i.e.}, 56 bits (1 for the MAC subheader and 6 for the actual RAR message) \cite[Section 6.1.5]{3G8}. There can be an optional additional byte containing a backoff indicator for all RARs. Let $M_{cs}$ denote the effective channel code rate chosen. Assuming that the RAR messages are sent only to the devices whose preambles were successfully detected (\emph{i.e.}, no Backoff Indicator is included in RAR message), the number of physical bits $L_b$ required to transmit one RAR message is equal to $\lceil 56/M_{cs} \rceil$ bits, where, $\lceil . \rceil$ denotes the ceil function. Let the number of symbols in a PRB that carry data bits be denoted by $D_s$. Then $N_{rar}^{prb}$, which denotes the number of RAR messages that can be included in one PRB, will be equal to $2D_sM_{cs}/56$, since QPSK is used. The number of PRBs, \emph{i.e.}, $N_{prb}$, in a subframe is equal to $B (MHz)/180 (kHz)$. Hence, denoting the maximum number of RAR messages that can be sent in one subframe 
by $N_{max}^{rar}$, the following relation holds,
\begin{equation}
\label{max-rars}
 N_{max}^{rar} = min \left( \frac{2D_sM_{cs}N_{prb}}{56}, 39 \right),
\end{equation}since, the maximum number of RAR messages that can be constructed with 2216 bits is equal to $\frac{2216}{56} \approx 39$.
\subsection{Energy Consumption Model for RAR Messages}
A well-accepted model of the energy consumption in a typical LTE-A eNodeB is provided in \cite{frenger, rohit}. The energy consumption of a eNodeB increases linearly with the utilization of the power amplifier (no. of PRB pairs on which data is scheduled in the downlink). Since a TTI lasts for 1 ms in LTE-A, the  energy consumed to transmit data on $k$ PRB pairs in a TTI can be computed as follows \cite{frenger, rohit}:
\begin{equation} 
\label{powerinput}
 E_{c}^{k} = \left(P_o + k\delta P_{t}\right)10^{-3}\hspace{2pt} \text{J},
\end{equation}where, $\delta$ is the power amplifier efficiency and  $P_{t}$ is the power transmitted per PRB and $P_o$ is a constant power factor signifying a fixed power consumption part of the power amplifier. Since the eNodeB practices Adaptive Modulation and Coding in LTE-A, we assume that the power transmitted per PRB is fixed and equal to $P_{t}$.  
In order to investigate the energy consumed by the eNodeB in serving the IoT devices RACH requests, we need to consider the downlink transmissions made by the eNodeB to serve the devices RACH requests which includes all the RAR messages sent in response to successful preamble transmissions. Now, if preambles from $r$ devices are successful in a RACH slot, then the number of subframes it takes to send RAR messages to all the successful devices would be given by:
\begin{equation} 
\label{no-of-subframes-for-rar}
N_{sf} = \lceil \frac{r}{N_{max}^{rar}} \rceil.
\end{equation} Denoting the size of the RAR window by $W_{rar}$, we can then express the energy consumed in joules by the eNodeB to send RAR messages to these $r$ successful devices, \emph{i.e.,} $E_{in}^{r}$, as,
\begin{equation} 
\small
\label{powerinput-rar}
 E_{in}^{r} = \sum_{j=1}^{\mathclap{\substack{min(N_{sf},W_{rar})}}}\left(P_o + \frac{min\left(r-(j-1)N_{max}^{rar},N_{max}^{rar}\right)56}{2D_sM_{cs}}\alpha P_{t}\right)10^{-3}.
\end{equation}

In \eqref{powerinput-rar}, the upper limit in the summation indicates that the number of subframes used to send RAR messages cannot exceed the RAR window size $W_{rar}$. Therefore, if $r$ is large enough for $N_{sf}$ to exceed $W_{rar}$, only as many RAR messages are sent as can be included in $W_{rar}$ subframes of the RAR window. The rest of the devices will back off assuming their preambles collided and return back to attempt again. The term within the summation gives the energy consumed in each subframe which depends on the number of RAR messages sent. The term $r-(j-1)N_{max}^{rar}$ indicates the number of devices that are yet to receive RAR messages in the $j^{th}$ subframe within the RAR window. Hence, the term $min\left(r-(j-1)N_{max}^{rar},N_{max}^{rar}\right)$ points out that the number of devices which receive the RAR messages in the $j^{th}$ subframe of the RAR window is $r-(j-1)N_{max}^{rar}$ if $r-(j-1)N_{max}^{rar} < N_{max}^{rar}$, else the number is limited to  $N_{max}^{rar}$.
\section{Performance Metrics and Analysis of EAB-BF}\label{analysis-section}
\begin{table}
\footnotesize
\centering
\footnotesize \caption{SUMMARY OF THE SYMBOLS USED.}
\begin{tabular}{ ||p{0.8cm}|p{7.2cm}||}
 \hline
 \textbf{Symbol} & \textbf{Parameters}
 \\
 \hline
 $F_i$ & Number of devices activated in the interval $i$
 \\
 \hline
 $\varGamma_i$  & Total number of devices (new and retransmissions) that arrive in RACH slot $i$
\\
\hline
$A_i$ & Number of devices that pass the EAB test and make a preamble transmission in RACH slot $i$
\\
\hline
$\varGamma^{'}_{i}$ & Number of devices which fail in the EAB test in RACH slot $i$
\\
\hline
$S_i$ &  Number of devices that are successful in their preamble transmission in slot $i$
\\
\hline
$C_i$& Number of devices that collide in slot $i$
\\
\hline
$\varUpsilon_i$& Number of devices successful in preamble transmission till slot $i$
\\
\hline
\end{tabular}
\label{table-for-analysis}
\end{table}
Having obtained the energy consumption model of the eNodeB for RAR messages, we now analyze EAB-BF to optimize its settings so that minimum eNodeB energy is consumed to provide the desired QoS (Success Probability and Mean Access Delay) to the devices. For a random variable $X$, we use $\mathbb{E}\left[X\right]$ and $\overline{X}$  interchangeably throughout this paper to denote the mean of $X$. The probability of an event $A$ is denoted by P($A$). Similarly, the conditional expection of $X$ given $A$ is denoted by $\mathbb{E}\left[X|A\right]$ and P($B|A$) denotes the conditional probability of $B$ given $A$.

The following notations are used and summarized in Table \ref{table-for-analysis}. Let $F_i$ represent the number of devices activated in the interval $i$. Then $\mathbb{E}\left[F_i\right]$ represents the access intensity given by \eqref{access-intensity-eqn}. Let $\varGamma_i$ denote the total number of devices (new and collided) that arrive in slot $i$ to make RACH attempts. Out of  $\varGamma_i$ that arrive, the number of devices that pass the EAB test and make a preamble transmission in the $i^{th}$ slot be denoted by $A_i$. Let $\varGamma^{'}_{i}$ denote the number of devices that fail in the EAB test in slot $i$ and back off for time $T_{eab}$. Out of the $A_i$ that attempt, let $S_i$ denote the number of devices that succeed in the $i^{th}$ slot while $C_i$ denotes the number of devices that collide in slot $i$. Let $\varUpsilon_i$ denote the number of devices that have cleared RACH successfully till slot $i$.  

\subsection{Performance Metrics}
In essence, if we can obtain the expected number of RACH attempts $\mathbb{E}\left[A_i\right]$ and collisions $\mathbb{E}\left[C_i\right]$ for slot $i$, and the expected number of successes $\mathbb{E}\left[\varUpsilon_i\right]$  till slot $i$, then with $N$ devices in the cell, the collision probability denoted by $P_C$, success probability denoted by $P_S$ and the mean of the access delay denoted by $T_{AD}$ can be obtained as follows \cite{wali}:
\subsubsection{Success Probability and Mean Number of Attempts}
A device makes a preamble transmission on passing the EAB test and because we have assumed that a non-colliding preamble implies successful RACH completion, the preamble success probability denoted by $P_S$ can be obtained as:
\begin{equation}
 P_S = 1 - P_C = 1 - \frac{\sum\limits_{i:\mathbb{E}\left[\varUpsilon_i\right] < N} \mathbb{E}\left[C_i\right]}{\sum\limits_{i:\mathbb{E}\left[\varUpsilon_i\right] < N} \mathbb{E}\left[A_i\right]}.
\end{equation}
Since each RACH attempt is independent for a device, the number of attempts to clear the RACH is a goemetric random variable. Denoting the number of attempts to successfully clear the RACH procedure as $N_A$, the expected number of attempts $\mathbb{E}\left[N_A\right]$ can be obtained as:
\begin{equation}
\label{ea}
\mathbb{E}\left[N_A\right] = \frac{1}{P_S} = \frac{\sum\limits_{i:\mathbb{E}\left[\varUpsilon_i\right] < N} \mathbb{E}\left[A_i\right]}{\sum\limits_{i:\mathbb{E}\left[\varUpsilon_i\right] < N} \left(\mathbb{E}\left[A_i\right]-\mathbb{E}\left[C_i\right]\right)}.
\end{equation}
\subsubsection{Mean Access Delay}
The mean access delay is the mean delay incurred by a device to successfully complete the RACH procedure starting from the time of its activation. It includes the backoff delays because of EAB failures and the RACH failures.

Clearly, the probability that a device passes the EAB test is $P_{eab}$. Also, the number of attempts needed to pass the EAB test is a geometric random variable since each attempt is independent. Therefore, the mean number of attempts before a device passes the EAB test is $(1-P_{eab})/P_{eab}$. But it requires mean number of attempts $\mathbb{E}\left[N_A\right]$ to succeed in RACH itself as in \eqref{ea}. Let us assume that for each failed RACH attempt, it takes a time $T_R$ to realize there was a collision and then a  backoff time $T_B$ before the device starts trying to pass EAB test again for attempting RACH. Note that $T_R$ is considered because a device learns about the collision after not receiving the RAR message (since we ignore the \emph{capture effect}). Since a device backs off for $T_{eab}$ whenever it fails to pass the EAB test and since the 
$N_A^{th}$ RACH attempt is a success, ignoring the small duration between its activation and the next RACH opportunity that appears, we can obtain the expected value of the access delay $T_{AD}$ for a device as:
\begin{multline*}\mathbb{E}\left[T_{AD}\right] = \mathbb{E}\left[T_{rach}\right] + \left[\frac{(1-P_{eab})T_{eab}}{P_{eab}} \mathbb{E}\left[N_A\right]\right]\\
+ \left[\left(\mathbb{E}\left[N_A\right]-1\right) \left(\mathbb{E}\left[T_R\right] + \mathbb{E}\left[T_B\right]\right)\right],
\end{multline*}where, $T_{rach}$ is the time for the device to get the contention resolution message in the RACH attempt that is successful. The mean backoff time after collision is $W/2$ since the backoff time is uniformly distributed over $\left[0,W-1\right]$ window. Hence, 
\begin{multline}\label{ed}
\mathbb{E}\left[T_{AD}\right] = \mathbb{E}\left[T_{rach}\right] + \left[\frac{(1-P_{eab})T_{eab}}{P_{eab}} \mathbb{E}\left[N_A\right]\right]\\
+ \left[\left(\mathbb{E}\left[N_A\right]-1\right) \left(\mathbb{E}\left[T_R\right] + W/2\right)\right].
\end{multline}
\subsubsection{Energy Per Cycle}
Since our goal is to minimize the eNodeB energy consumption also, we obtain the mean eNodeB energy consumed per cycle in joules  denoted by $\overline E_{c}^{cyl}$ in serving the RACH requests of the IoT devices as:
\begin{equation}
\footnotesize
\label{eqn-for-p-in-avg}
  \overline E_{c}^{cyl}=\sum\limits_{i:\mathbb{E}\left[\varUpsilon_i\right] < N} \sum_{j=1}^{\mathclap{\substack{min(N_{sf},W_{rar})}}}                 
        P_o + \frac{min(\mathbb{E}\left[S_i\right]-(j-1)N_{max}^{rar},N_{max}^{rar})56\alpha P_{t}}{2D_sM_{cs}}10^{-3}. 
\end{equation}
The inner sum in \eqref{eqn-for-p-in-avg} evaluates the mean energy consumed to transmit RAR messages to $\mathbb{E}\left[S_i\right]$ devices which are successful in RACH slot $i$ of the access cycle. This is essentially \eqref{powerinput-rar} with $r$ replaced by $\mathbb{E}\left[S_i\right]$. The outer sum covers all RACH slots till all the devices are successful in the cycle.
\subsection{Analysis of EAB-BF}\label{analysis}
Clearly, to obtain both the QoS ($P_S$ and $\mathbb{E}\left[T_{AD}\right]$) provided by EAB-BF and the mean energy per cycle $\overline E_{c}^{cyl}$, the terms $\mathbb{E}\left[A_{i}\right]$, $\mathbb{E}\left[S_{i}\right]$ and $\mathbb{E}\left[C_{i}\right]$ need to be computed. We next analyze EAB-BF to obtain closed form expressions for these quantities. With the notations defined earlier, we have the following results:
\begin{claim}
  The expected number of devices that make RACH attempt (preamble transmission) in slot $i$ is given by the following expression:
     \begin{multline}\label{final-eqn-for-Ai}
     \mathbb{E}\left[A_{i}\right] = \sum_{j=0}^{Q} P_{eab}(1-P_{eab})^j\left[\mathbb{E}\left[F_{i-jq}\right] + \sum_{\mathclap{\substack{l=i-(jq+\frac{W}{r_p})}}}^{i-(jq+1)} \frac{r_p\mathbb{E}\left[C_l\right]}{W}\right];\\      
      Qq+1\leq i\leq (Q+1)q, \hspace{5pt} Q \in \mathbb{Z}^+, q \triangleq \frac{T_{eab}}{r_p}.
     \end{multline}
\end{claim}
\begin{proof}
The proof is shown in Appendix \ref{proof1}.
\end{proof}
We now turn our attention to the expected number of collisions $\mathbb{E}\left[C_i\right]$. 
\begin{claim}
  The expected number of devices that collide in slot $i$ can be approximated by the following expression:
 \begin{equation}\label{final-eqn-for-Ci}
 \mathbb{E}\left[C_i\right] = \mathbb{E}\left[A_i\right] - \mathbb{E}\left[A_i\right]\left(1-\frac{1}{K}\right)^{\mathbb{E}\left[A_i\right]-1}.
 \end{equation}
\end{claim}
\begin{proof}
The proof is shown in Appendix \ref{proof2}.
\end{proof}

It will be shown that the approximations made to obtain a closed form expression for $\mathbb{E}\left[C_i\right]$ are valid and yield accurate results. The evolution of the system is now completely expressed in terms of expected values given by \eqref{final-eqn-for-Ai} and \eqref{final-eqn-for-Ci}. These are iterative equations, \emph{i.e.}, starting with $\mathbb{E}\left[A_1\right] = P_{eab} \times \mathbb{E}\left[F_1\right]$ (since $P_{eab}$ is known and $\mathbb{E}\left[F_i\right]$ can be obtained from \eqref{access-intensity-eqn}) we can obtain $\mathbb{E}\left[C_1\right]$. Then we can find $\mathbb{E}\left[A_2\right]$ by using $\mathbb{E}\left[C_1\right]$ along with $\mathbb{E}\left[F_2\right]$ which in turn helps to obtain $\mathbb{E}\left[C_2\right]$ and so on.

\subsection{Analytical and Simulation Results}\label{simulation-results}
\begin{figure}
\centering
\mbox{\subfigure[EAB$\left(0.7,8\hspace{2pt}s\right)$]{\includegraphics[width=1.75in]{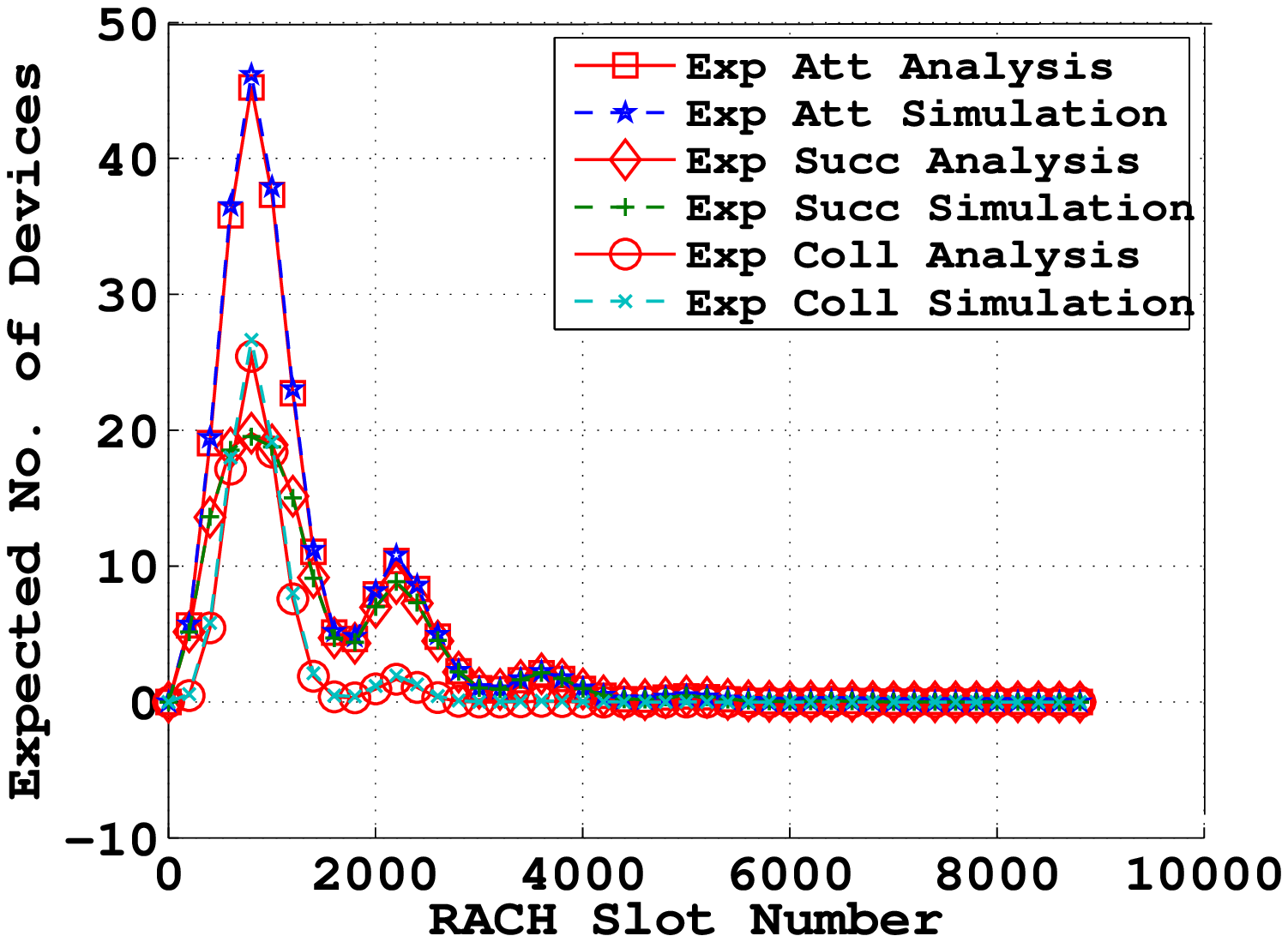}}\quad
\subfigure[EAB$\left(0.08,0.5\hspace{2pt}s\right)$]{\includegraphics[width=1.75in]{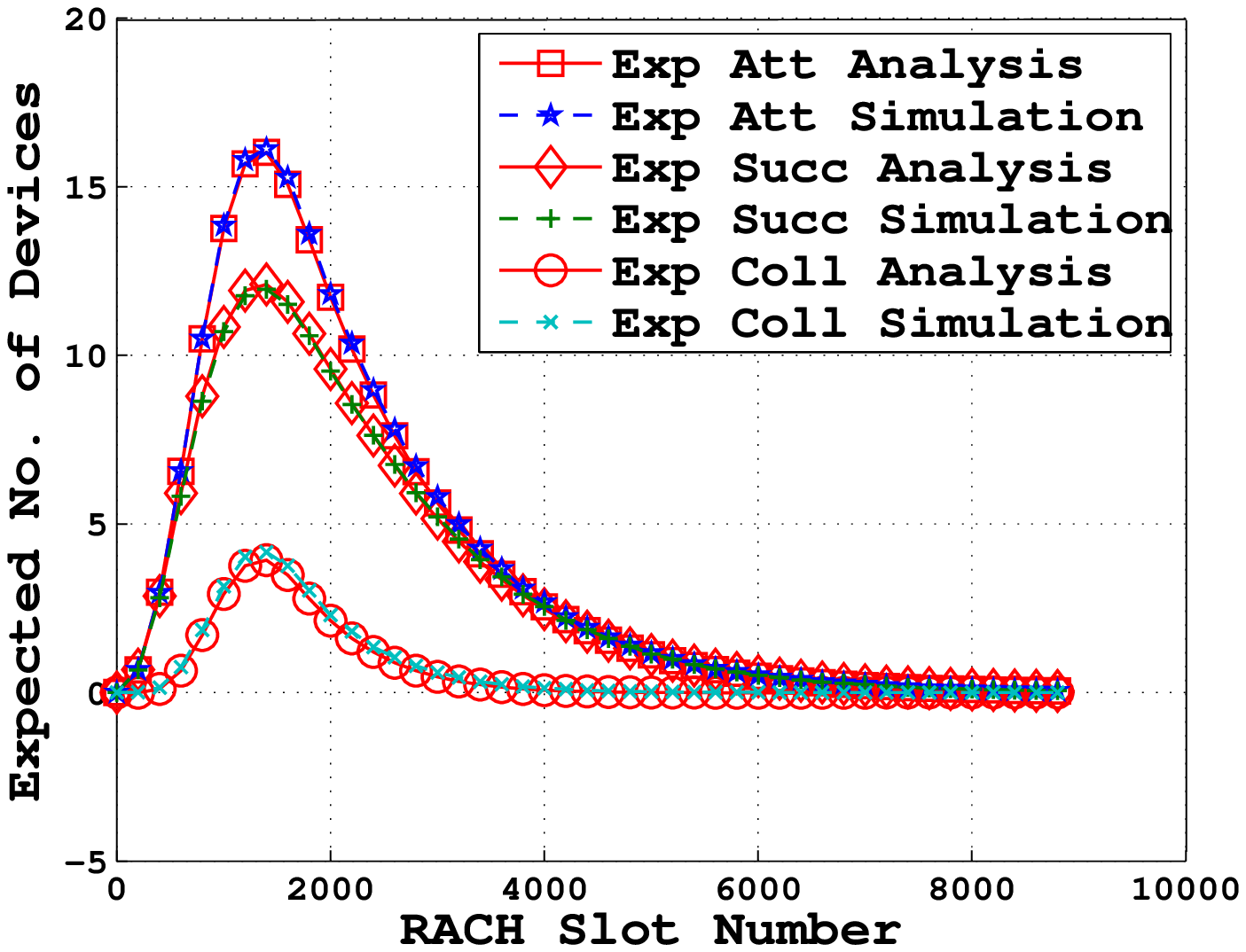} }}
\caption{Analysis and Simulation plots for expected attempts, successes and collisions V/S RACH slots}
\label{analysis-verification}
\end{figure}
We now verify the accuracy of our analytical expressions using Monte Carlo simulations that average over 5000 samples, \emph{i.e.,} 5000 access cycles. The simulation code is written in MATLAB. A total of $N=30000$ devices are present in the cell. The activation times of these devices is set to follow the beta distribution as specified by 3GPP \cite{3G6} and given by \eqref{pt}. Values of $\alpha$ = 3 and $\beta$ = 4 are considered as provided in \cite{3G6}. Also, $T_A$ is assumed to be 10 s. 

The RACH parameters are set according to \cite[Table 6.2.2.1.1]{3G6}. The backoff window is set to 20 ms that is uniformly distributed, so  $\mathbb{E}\left[T_B\right]=10$ ms. The RAR Response window size is 5 ms. Since we do not consider the \emph{capture effect}, any device which does not receive the RAR response in at most 5 ms duration after transmitting a preamble concludes its preamble has collided and decides to back off. So $\mathbb{E}\left[T_R\right] = 5$ ms\footnote{Note that if we had considered the \emph{capture effect}, the device would learn of its collision after the Contention Resolution window. Then $\mathbb{E}\left[T_R\right]$ would become 25 ms resulting in negligible changes to the results obtained without \emph{capture effect}.}. The contention resolution message window is uniformly distributed over 48 ms interval. Hence, $\mathbb{E}\left[T_{rach}\right]=24$ ms. PRACH configuration index is 6, \emph{i.e.,} $r_p=5$ ms. 

We run the simulation for EAB$\left(0.08,0.5\hspace{2pt}s\right)$ along with the three 3GPP proposed settings during standardization. We show the plots for the expected number of devices that attempt, succeed and collide in each RACH slot  for EAB$\left(0.7,8\hspace{2pt}s\right)$ and EAB$\left(0.08,0.5\hspace{2pt}s\right)$ as shown in figures \ref{analysis-verification}(a) and \ref{analysis-verification}(b). It can be seen that the analysis and the simulation plots match very well verifying the accuracy of our analytical model and validating the approximations that were done. We also tabulate the values for success probability $P_S$ and mean access delay $\mathbb{E}\left[T_{AD}\right]$ for successful RACH completion in Table \ref{table:3} for 
all the settings. The analysis and simulation values are seen to match well again.  
\begin{table}
\scriptsize
\centering
\caption{COMPARISON OF ANALYSIS AND SIMULATION RESULTS. SIMULATION RESULTS ARE SHOWN IN BRACKETS}
\begin{tabular}{ ||p{2.2cm}|p{2.0cm}|p{2.0cm}||}
 \hline
   & $P_S$ & $\mathbb{E}\left[T_{AD}\right]$
 \\
 \hline
  \textbf{EAB$\left(0.5,16\hspace{2pt}s\right)$}  & 0.84 (0.82) & 19.00 s (19.10 s)
  \\
 \hline
 \textbf{EAB$\left(0.7,8\hspace{2pt}s\right)$}  & 0.75 (0.72) & 4.56 s (4.70 s)
 \\
\hline
\textbf{EAB$\left(0.9,4\hspace{2pt}s\right)$} & 0.11 (0.13) & 4.25 s (3.55 s)
\\
\hline
\textbf{EAB$\left(0.08,0.5\hspace{2pt}s\right)$} & 0.84 (0.89) & 6.80 s (6.46 s)
\\
\hline
\end{tabular}
\label{table:3}
\end{table}
\vspace{-0.1in}
\section{Joint Optimization of EAB-BF and eNodeB Energy Consumption for Energy Efficient Optimal Performance}\label{optimization-section}
\label{power-minimize-section}
We now use our validated analytical model to investigate the optimal performance of EAB-BF jointly with respect to eNodeB energy consumption. The following notations are used. Let $\left(P_{eab}^*, T_{eab}^*\right)$ denote the optimal setting of EAB-BF that minimizes the mean eNodeB energy consumption per cycle  $\overline E_{c}^{cyl}$ under the constraint that the success probability $P_S$ is greater than some threshold $P_{min}$ and mean access delay $\mathbb{E}\left[T_{AD}\right]$ is less than some threshold $T_{max}$. We denote $\overline E_{c}^{cyl}$ achieved at $\left(P_{eab}^*, T_{eab}^*\right)$ as $\overline E_{c}^{min}$, $P_S$ as $P_{S}^{max}$ and $\mathbb{E}\left[T_{AD}\right]$ as $\overline T_{AD}^{min}$. This investigation helps us achieve the following objectives:

\emph{(i)} To show how our analytical expressions can be used to build search algorithms to obtain $\left(P_{eab}^*, T_{eab}^*\right)$ settings. These settings can be used by network operators to guarantee the desired QoS to IoT devices with minimal eNodeB energy consumption if EAB-BF is employed.

\emph{(ii)} To show that the optimal performance of EAB-BF is better than EAB-BB, through results obtained from one example search algorithm presented.

\emph{(iii)} To show that the settings that 3GPP considered for EAB-BF during standardization process not only result in sub-optimal performance but result in excessive eNodeB energy consumption. 

\emph{(iii)} To make corrections to sub-optimal 3GPP settings that result in substantial gains in the performance.

To achieve the above objectives, we first formulate  a constrained optimization problem. Next, we present an example search algorithm that makes use of our analytical expressions and some approximations to obtain the near-optimal settings which are denoted as $\left(\widehat P_{eab}, \widehat T_{eab}\right)$ to distinguish from the optimal $\left(P_{eab}^*, T_{eab}^*\right)$ settings.

\begin{table*}
\scriptsize
\centering
\caption{OPTIMIZING THE PERFORMANCE OF EAB-BF FOR VARIOUS LOWER BOUNDS ON SUCCESS PROBABILITY}
\begin{tabular}{ ||p{0.9cm}|p{1.2cm}|p{1.1cm}|p{1.1cm}|p{2.2cm}|||p{0.9cm}|p{1.2cm}|p{1.1cm}|p{1.1cm}|p{2.2cm}||  }
 \hline
   $P_{min}$ &  $\overline T_{AD}^{min}$ & $P_{S}^{max}$ & $\overline E_{c}^{min}$ & $\left(\widehat P_{eab}, \widehat T_{eab}\right)$ &  $P_{min}$ & $\overline T_{AD}^{min}$&  $P_{S}^{max}$ & $\overline E_{c}^{min}$& $\left(\widehat P_{eab}, \widehat T_{eab}\right)$ 
 \\
 \hline
 0.05&1.49 s&0.3384&0.45 KJ & $\left(0.17, 0.10\hspace{2pt}s\right)$&0.55&1.93 s& 0.5528&0.54 KJ& $\left(0.16,0.20\hspace{2pt}s\right)$
 \\
 \hline
0.10&1.49 s&0.3384& 0.45 KJ& $\left(0.17, 0.10\hspace{2pt}s\right)$&0.60&2.25 s& 0.6032& 0.60 KJ& $\left(0.13,0.20\hspace{2pt}s\right)$
 \\
 \hline
0.15&1.49 s&0.3384& 0.45 KJ & $\left(0.17, 0.10\hspace{2pt}s\right)$&0.65&2.63 s& 0.6525& 0.67 KJ&$\left(0.15,0.30\hspace{2pt}s\right)$
 \\
 \hline
0.20&1.49 s&0.3384&0.45 KJ& $\left(0.17, 0.10\hspace{2pt}s\right)$&0.70&3.29 s& 0.7047& 0.77 KJ& $\left(0.08,0.20\hspace{2pt}s\right)$
 \\
 \hline
0.25&1.49 s&0.3384& 0.45 KJ&$\left(0.17, 0.10\hspace{2pt}s\right)$&0.75&3.80 s& 0.7505& 0.92 KJ& $\left(0.22,0.80\hspace{2pt}s\right)$
 \\
 \hline
0.30&1.49 s&0.3384& 0.45 KJ&$\left(0.17, 0.10\hspace{2pt}s\right)$&0.80&4.89 s& 0.8011& 1.14 KJ& $\left(0.22,1.10\hspace{2pt}s\right)$
 \\
 \hline
0.35&1.58 s& 0.3648& 0.45 KJ&$\left(0.16,0.10\hspace{2pt}s\right)$&0.85&6.91 s& 0.8511& 1.49 KJ& $\left(0.12,0.80\hspace{2pt}s\right)$
 \\
 \hline
0.40&1.61 s& 0.4152&0.46 KJ& $\left(0.14,0.10\hspace{2pt}s\right)$&0.90&10.46 s& 0.9002&2.15 KJ& $\left(0.06,0.60\hspace{2pt}s\right)$
 \\
 \hline
0.45&1.64 s& 0.4642&0.48 KJ& $\left(0.12,0.10\hspace{2pt}s\right)$&0.95&19.14 s& 0.9516&3.48 KJ& $\left(0.09,1.80\hspace{2pt}s\right)$
\\
\hline
0.50&1.72 s& 0.5059&0.51 KJ& $\left(0.19,0.20\hspace{2pt}s\right)$&1.00&19.14 s& 0.9516&3.48 KJ& $\left(0.09,1.80\hspace{2pt}s\right)$
\\
\hline
\end{tabular}
\label{table:4}
\end{table*}

Even though the mean energy consumed per cycle $\overline E_{c}^{cyl}$ is a function of the variables $W$, $r_p$, $P_{eab}$ and $T_{eab}$, 
our interest currently lies only in the way $P_{eab}$ and $T_{eab}$ affect the performance of EAB-BF. Hence, we choose to minimize the objective $\overline E_{c}^{cyl}$ with respect to $P_{eab}$ and $T_{eab}$  under the constraint of a lower bound on the success probability $P_S$ and an upper bound on the mean access delay $\mathbb{E}\left[T_{AD}\right]$.  Hence, the optimization is formulated as follows:
\begin{equation}\label{opti-formulation}
\begin{aligned}
& \underset{P_{eab},T_{eab}}{\text{min}}
& & \overline E_{c}^{cyl}\\
& \text{s.t.} & &  P_S \geq P_{min}, \mathbb{E}\left[T_{AD}\right] \leq T_{max}\\ 
& & &P_{eab}, P_S,P_{min} \in \left[0,1\right]\\
& & & W,r_p,T_{eab} \in \mathbb{Z}^+,\\
\end{aligned}
\end{equation}where, $P_{min}$ is the lower bound on the success probability, $T_{max}$ the upper bound on the mean access delay and  $W$, $r_p$ are expressed in milliseconds and $T_{eab}$ in seconds. The objective $\overline E_{c}^{cyl}$ is given by \eqref{eqn-for-p-in-avg}.

To present an example search algorithm to obtain optimal settings, we make the following observations. Note that $\mathbb{E}\left[S_i\right] = \mathbb{E}\left[A_i\right] - \mathbb{E}\left[C_i\right]$ and therefore, the objective in \eqref{opti-formulation} is non linear in $P_{eab}$ and $T_{eab}$. Also, with $P_{eab} \in [0,1] \subseteq \mathbb{R}$ and $T_{eab} \in \mathbb{Z}^+$, this formulation belongs to the class of Mixed Integer Non Linear Programming problems (MINLP) which are in general known to be hard to solve theoretically \cite{minlp}. Usually, some form of single-tree and multi-tree search methods are applied for solving MINLP problems. For convex MINLP, hybrid methods that combine the strengths of single-tree and multi-tree search methods are used. But to classify the objective and/or the feasible region as convex in \eqref{opti-formulation}, $T_{eab}$  needs to be relaxed to a continuum-valued variable. But  that would make $\mathbb{E}\left[A_i\right]$ (on which $\mathbb{E}\left[S_i\right]$ 
depends) incomputable by 
\eqref{final-eqn-for-Ai} (since $T_{eab}$ appears in the discrete summation). Therefore, the objective and/or the feasible region cannot be approximated as convex \cite{boyd}. Moreover $\mathbb{E}\left[A_i\right]$,  on which the objective function depends, is multi-modal in nature as shown in Fig. \ref{ai-beh} for $i = 5000, 15000$ as $P_{eab}$ and $T_{eab}$ are varied. Even for a fixed $i$, the behavior of $\mathbb{E}\left[A_i\right]$ is not easy to characterize mathematically. The fact that $\mathbb{E}\left[C_i\right]$ is a 
non linear function of $\mathbb{E}\left[A_i\right]$ as seen in \eqref{final-eqn-for-Ci} adds to the difficulty.

For the class of nonconvex MINLP problems, branch-and-bound (BB) (also called  spatial BB (sBB)) is one of the well-known methods \cite{minlp}. But the BB algorithm requires \emph{(i)} a procedure to compute a lower bound on the optimal objective function value of a subproblem and \emph{(ii)} a procedure for partitioning the feasible set of a subproblem. But even if we sample $P_{eab}$ and $T_{eab}$ and attempt enumeration in our formulation, the requirements of the BB algorithm cannot be satisfied because of the way $\mathbb{E}\left[A_i\right]$ depends on $P_{eab}$ and $T_{eab}$.

Because of the reasons mentioned above, we present an algorithm that employs exhaustive search. To this end, we evaluate the objective at only finite sampled set of $P_{eab}$ and $T_{eab}$ values. We then search for the combination $\left(P_{eab},T_{eab}\right)$ that results in the minimum mean energy consumption $\overline E_{c}^{min}$ while satisfying the constraint of $P_{min}$ and $T_{max}$. The values of the objective obtained at each sampled $\left(P_{eab},T_{eab}\right)$ combination are accurate since they are computed using exact equations (\eqref{eqn-for-p-in-avg}, \eqref{final-eqn-for-Ai} and \eqref{final-eqn-for-Ci}) that have been derived. Since we have exact equations for $\overline E_{c}^{cyl}$ along with $P_S$ and $\mathbb{E}\left[T_{AD}\right]$ as derived in previous 
sections, the search is made simple and does 
not encounter any convergence issues that optimization algorithms usually face. Therefore, if the set of $\left(P_{eab},T_{eab}\right)$ combinations that have been sampled contains the minimizer of the objective, the solution is guaranteed to be optimal. On the other hand, if the set misses the minimizer of the objective function, then the solution obtained is near-optimal. This is only due to sampling at lesser resolution and not due to any fault in the method per se. Hence, the solution is at worst near-optimal and denoted by $\left(\widehat P_{eab}, \widehat T_{eab}\right)$. Note that at $\left(\widehat P_{eab}, \widehat T_{eab}\right)$, the success probability denoted by $P_{S}^{max}$ is the maximum that is possible and the mean delay denoted by $\overline T_{AD}^{min}$ is the minimum that is possible. 

Our example algorithm to search for $\left(\widehat P_{eab}, \widehat T_{eab}\right)$ that yields $\overline E_{c}^{min}$, $P_S^{max}$ and $\overline T_{AD}^{min}$ under the constraint of $P_{min}$ and $T_{max}$ is shown in Algorithm \ref{CHalgorithm}. The algorithm calculates $P_S$, $\mathbb{E}\left[T_{AD}\right]$ and $\overline E_{c}^{cyl}$ for each sampled pair $\left(P_{eab}, T_{eab}\right)$ in the \emph{for} loop between line 4 and 11. Within the \emph{for} loop, for each $\left(P_{eab}, T_{eab}\right)$, the \emph{while} loop from line 7 to 9 evaluates $\mathbb{E}\left[A_{i}\right]$, $\mathbb{E}\left[S_{i}\right]$, $\mathbb{E}\left[C_{i}\right]$, $\mathbb{E}\left[\varUpsilon_i\right]$ and $\overline E_{c}^{cyl}$ for each RACH slot $i$ using the closed form expressions obtained from our analysis. The \emph{Energy Consumption Minimizer} procedure then searches for the near-optimal pair $\left(\widehat P_{eab}, \widehat T_{eab}\right)$ that results in $\overline E_{c}^{min}$ for the constraints $P_{min}$ 
and $T_{max}$. The values of $P_S$ and $\mathbb{E}\left[T_{AD}\right]$ obtained at $\left(\widehat P_{eab}, \widehat T_{eab}\right)$ are then $P_{S}^{max} $ and $\overline T_{AD}^{min}$ respectively. 

\begin{algorithm}[h!]
\caption{To find $\overline E_{c}^{min}$, $\overline T_{AD}^{min}$, $P_{S}^{max}$, $\left(\widehat P_{eab}, \widehat T_{eab}\right)$}
\label{CHalgorithm}
\begin{algorithmic}[1]
\small
\State Initialize step sizes $\Delta P_{eab}$, $\Delta T_{eab}$, $\Delta P_{min}$. 
\State Initialize $N$, $r_p$, $W$, $K$, $\mathbb{E}\left[T_{rach}\right]$, $\mathbb{E}\left[T_{R}\right]$, $\mathbb{E}\left[T_{B}\right]$, $T_{max}$.
\Procedure{Performance\_Evaluator}{}
\For{$P_{eab}$ = 0 : $\Delta P_{eab}$ : 1}
\For{$T_{eab}$ = 0.1 : $\Delta T_{eab}$ : 20}
\State Initialize \scriptsize $\mathbb{E}\left[\varUpsilon_0\right],\overline E_{c}^{cyl}(P_{eab},T_{eab}) = 0$, $i = 1$, $q=\frac{T_{eab}}{r_p}$.
\small
\While{$\mathbb{E}\left[\varUpsilon_i\right] < N$}
\State Set $Q=\lfloor \frac{i}{q} \rfloor$ and evaluate:
\scriptsize
\begin{displaymath}
\begin{split}
 \hspace{56pt}&\mathbb{E}\left[A_{i}\right]=\sum_{j=0}^{Q} P_{eab}(1-P_{eab})^j\left[\mathbb{E}\left[F_{i-jq}\right] + \sum_{l=i-(jq+\frac{W}{r_p})}^{i-(jq+1)}\mathbb{E}\left[B_{l}\right]\right]\\
 &\mathbb{E}\left[C_i\right] =  \mathbb{E}\left[A_i\right]\left(1-\left(1-\frac{1}{K}\right)^{\mathbb{E}\left[A_i\right]-1}\right)\\
 &\mathbb{E}\left[S_i\right] = \mathbb{E}\left[A_i\right] - \mathbb{E}\left[C_i\right]\\
 &\mathbb{E}\left[\varUpsilon_i\right] = \mathbb{E}\left[\varUpsilon_{i-1}\right] + \mathbb{E}\left[S_i\right]\\
 &\overline E = 10^{-3}\sum_{j=1}^{\mathclap{\substack{min(N_{sf},W_{rar})}}}                 
        P_o + \frac{min(\mathbb{E}\left[S_i\right]-(j-1)N_{max}^{rar},N_{max}^{rar})56}{2D_sM_{cs}}\alpha P_{t}\\
 &\overline E_{c}^{cyl}(P_{eab},T_{eab}) = \overline E_{c}^{cyl}(P_{eab},T_{eab}) + \overline E\\
 &i = i+1
 \end{split}
\end{displaymath}
\small
\EndWhile
\scriptsize
\begin{displaymath}
\begin{split}
\hspace{40pt}&P_S(P_{eab},T_{eab}) = 1 - \frac{\sum\limits_{i:\mathbb{E}\left[\varUpsilon_i\right] < N} \mathbb{E}\left[C_i\right]}{\sum\limits_{i:\mathbb{E}\left[\varUpsilon_i\right] < N} \mathbb{E}\left[A_i\right]}; \hspace{5pt} \mathbb{E}\left[N_A\right] = \frac{1}{P_S}\\
&t\_v = \frac{(1-P_{eab})T_{eab}\mathbb{E}\left[N_A\right]}{P_{eab}} + \mathbb{E}\left[T_{rach}\right]\\
&\mathbb{E}\left[T_{AD}\right](P_{eab},T_{eab}) = t\_v + \left(\mathbb{E}\left[N_A\right]-1\right) \left(\mathbb{E}\left[T_R\right] + \mathbb{E}\left[T_B\right]\right)
\end{split}
\end{displaymath}
\small
\EndFor
\EndFor
\EndProcedure
\Procedure{Energy\_Consumption\_Minimizer}{}
\State Find $\overline E_{c}^{min}$ s.t $P_{S} \geq P_{min}, \hspace{5pt} \mathbb{E}\left[T_{AD}\right] \leq T_{max} $
\State Find $\left(\widehat P_{eab}, \widehat T_{eab}\right)$ for which $\overline E_{c}^{cyl} = \overline E_{c}^{min}$
\State Set $P_{S}^{max}=P_{S}$ and $\overline T_{AD}^{min} = \mathbb{E}\left[T_{AD}\right]$ at $\left(\widehat P_{eab}, \widehat T_{eab}\right)$
\EndProcedure
\end{algorithmic}
\end{algorithm}

\subsection{Comparison with EAB-BB}
We now compare the performance of EAB-BF obtained through our analysis with EAB-BB as provided in \cite{cheng3} for $N$=30000\footnote{We choose $N$=30000 since the case of 30000 devices with beta arrivals is considered the most severe, leading to RAN overload in a LTE-A cell.}. To do that, we implement our search algorithm in MATLAB. $P_{eab}$ is sampled in steps of 0.01 from 0 to 1, $T_{eab}$ is taken from 100 ms to 20 s in steps of 100 ms. The settings of other parameters are same as considered in Section \ref{simulation-results} for simulations that conform with 3GPP settings and the values used in \cite{cheng3}. $P_{min}$ is varied from 0 to 1 in steps of 0.01. $T_{max}$ is set to 50 seconds. We set $P_o = 170$ W, $P_t = 0.8$ W  and $\delta = 0.3$ as used in \cite{frenger}. We tabulate the results in Table \ref{table:4} for only a subset of $P_{min}$ values due to space constraints. The set of 3-tuples $\left(P_S^{max},\overline T_{AD}^{min},\overline E_{c}^{min}\right)$ obtained through this search 
provide the trade off in EAB-BF performance. Each 3-tuple indicates the near-optimal 
combination of $P_S$, $\mathbb{E}\left[T_{AD}\right]$ and $\overline E_{c}^{cyl}$ that can be achieved simultaneously when EAB-BF is employed. For comparison with EAB-BB,
 we consider only the set of near-optimal 2-tuples $\left(P_S^{max},\overline T_{AD}^{min}\right)$. We compare these results with the success probability and mean access delay pairs $\left(P_S,\overline D\right)$  that are obtainable through EAB-BB mechanism using various combinations of  transmission period of SIB14 ($T_S$) and the paging cycle ($T_P$), as shown through Fig. 6(a) and Fig. 6(c) in  \cite{cheng3}. 

For EAB-BF, from Table \ref{table:4}, it can be observed that the mean delay $\overline T_{AD}^{min}$ stays within 2 s till the success probability $P_{S}^{max}$ reaches 0.6  and remains within 5 s till $P_{S}^{max}$ reaches a value of 0.8. But for EAB-BB, the mean delay $\overline D$ increases almost monotonically to more than 5 s before the success probability $P_S$ reaches 0.5 and is almost 14 s at $P_S=0.9$ for EAB-BB as shown through Fig. 6(a) and Fig. 6(c) in \cite{cheng3}. Clearly, the mean access delay obtainable for EAB-BF is lesser than that of EAB-BB for the same success probability that can be achieved in both schemes. The mean access delay is only 10.46 s for the same success probability of 0.9 for EAB-BF proving its superiority over EAB-BB. 
\begin{figure*}
\centering
\mbox{\subfigure[Maximum Success Probability $P_{S}^{max}$]{\includegraphics[width=2.2in]{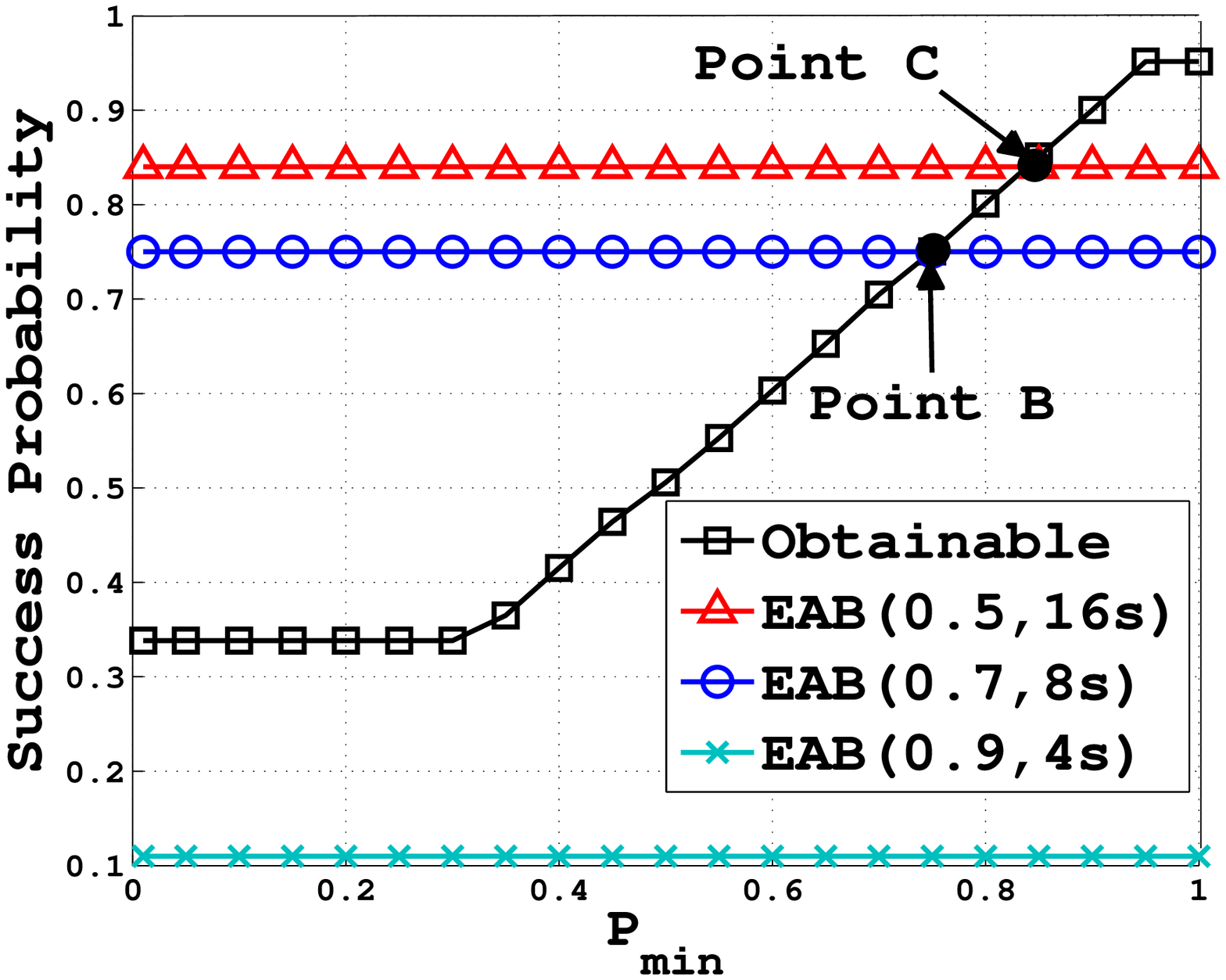}}\quad
\subfigure[Minimum Energy Consumption $\overline E_{c}^{min}$]{\includegraphics[width=2.2in]{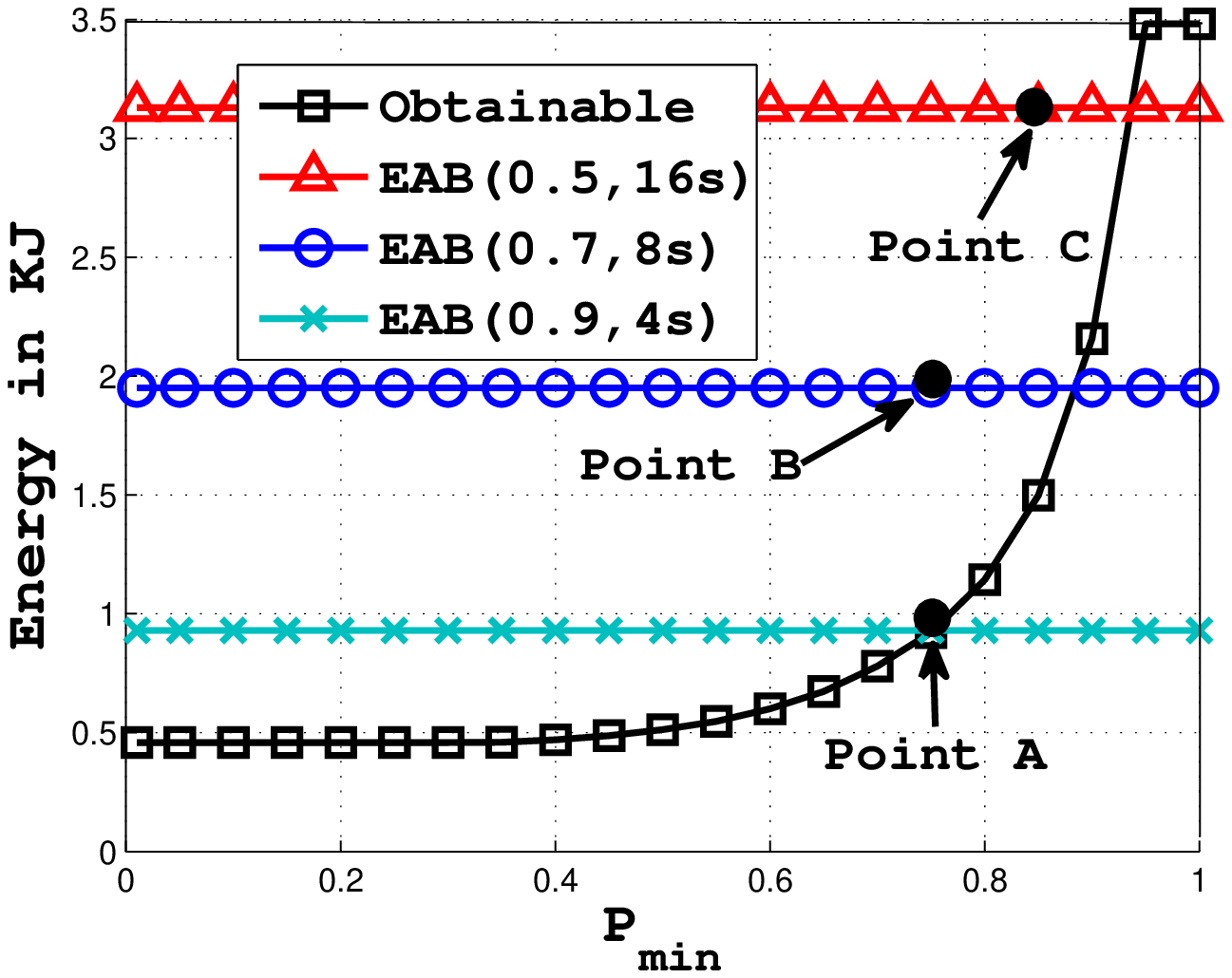}}\quad
\subfigure[Minimum Mean Access Delay $\overline T_{AD}^{min}$]{\includegraphics[width=2.2in]{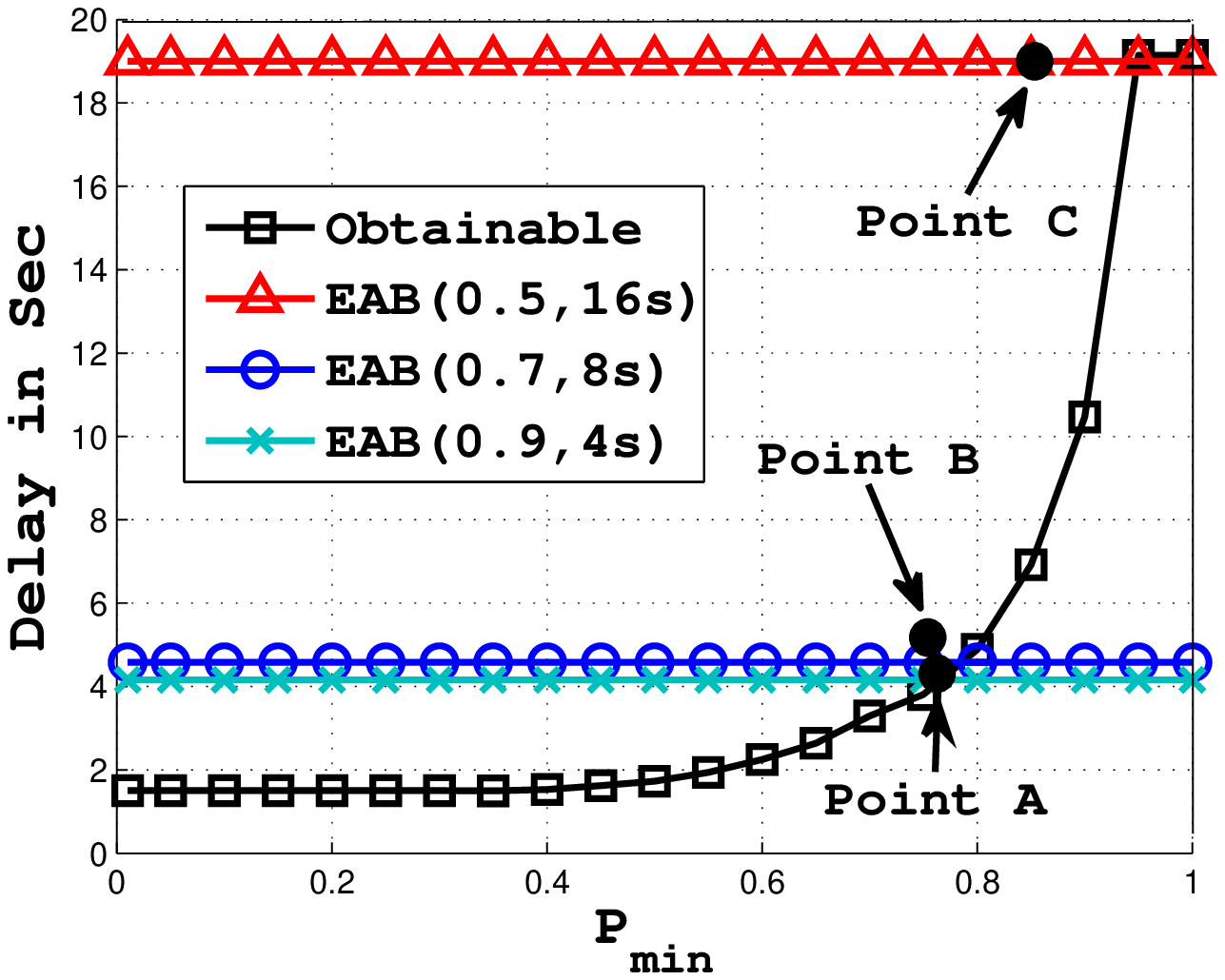}}}
\vspace{0pt}
\mbox{\subfigure[Gain in Success Probability]{\includegraphics[width=2.2in]{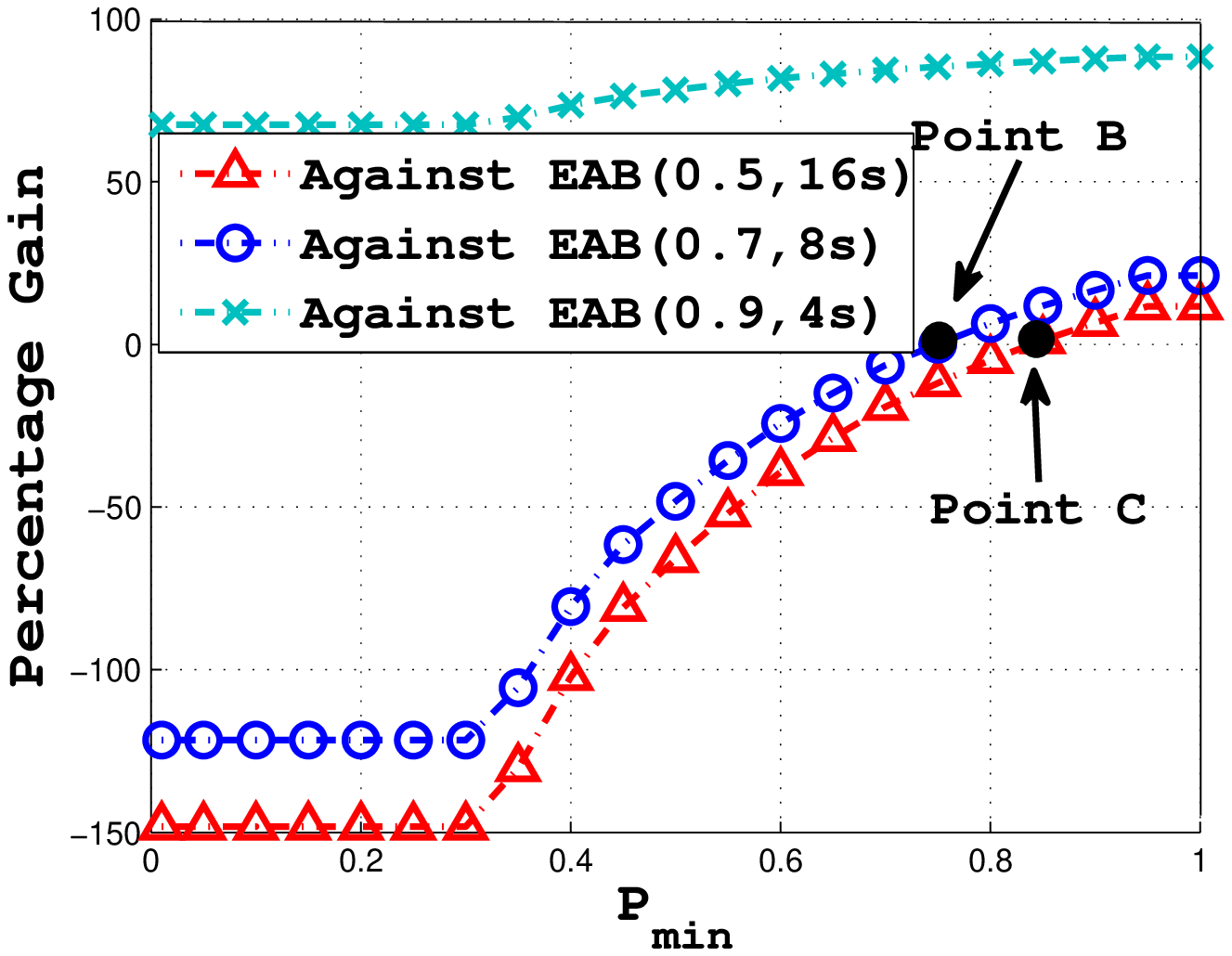}}\quad
\subfigure[Gain in Energy Consumption]{\includegraphics[width=2.2in]{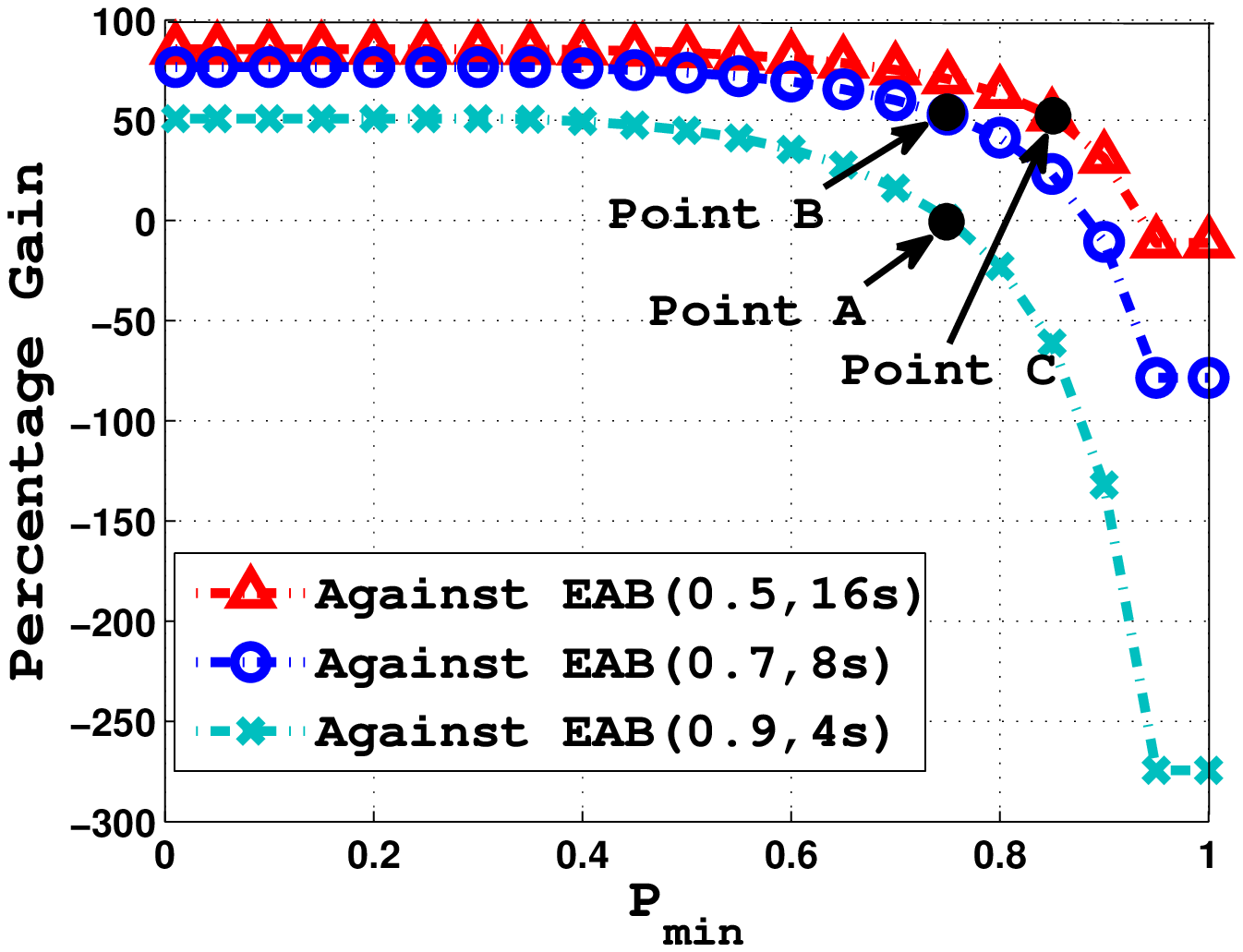}}\quad
\subfigure[Gain in Mean Access Delay]{\includegraphics[width=2.2in]{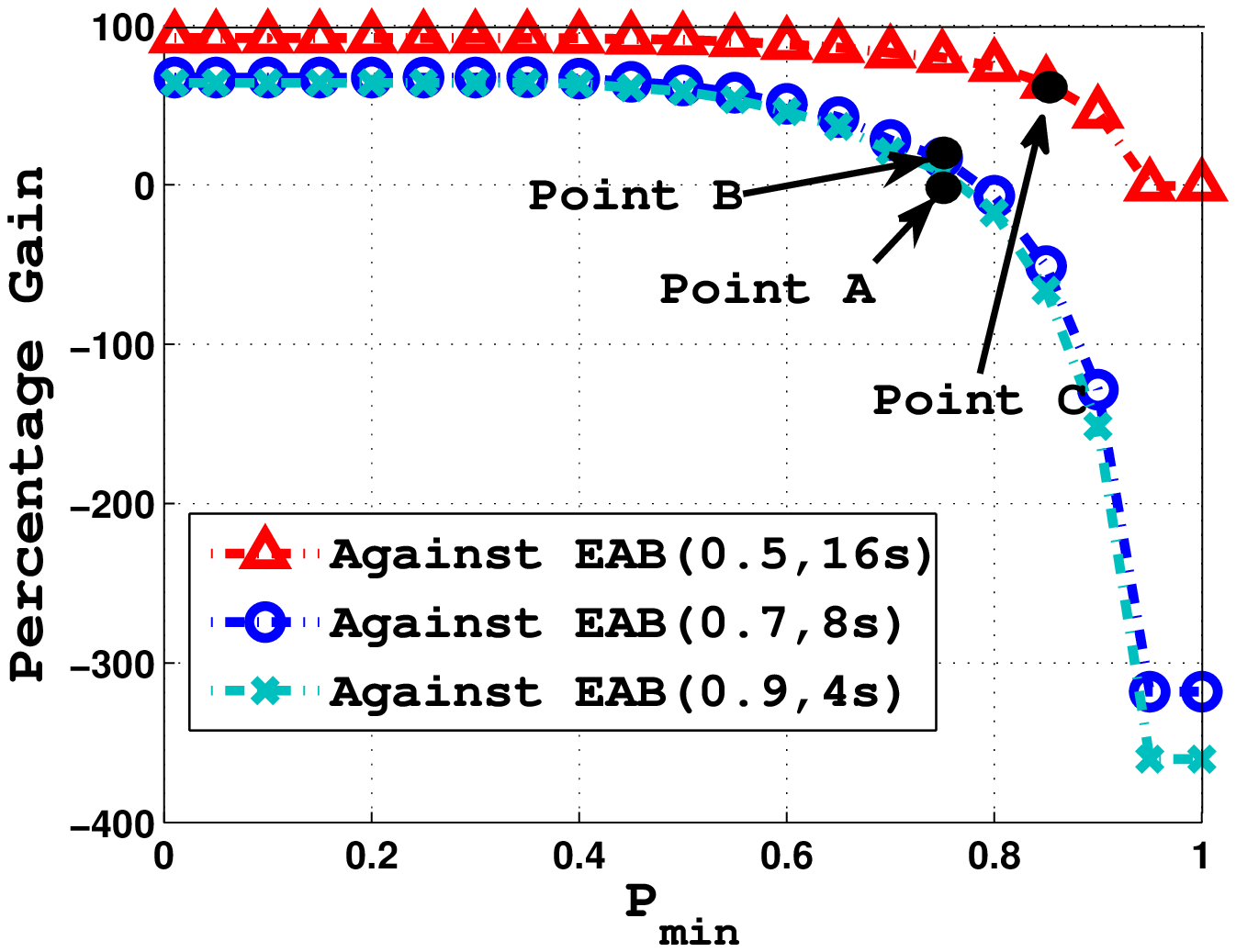}}}
\vspace*{-5mm}
\caption{Tradeoff and Gain Curves for maximum success probability $P_{S}^{max}$, minimum energy consumption $\overline E_{c}^{min}$, and minimum access delay $\overline T_{AD}^{min}$}
\label{optimize-together}
\end{figure*}
\begin{table}
\footnotesize
\centering
\footnotesize \caption{ENERGY CONSUMPTION PER CYCLE FOR 3GPP SETTINGS}
\begin{tabular}{ ||p{2.0cm}|p{5.0cm}||  }
 \hline
 \textbf{EAB Setting} & \textbf{Mean Energy Consumed Per cycle} \\
 \hline
EAB$\left(0.5,16\hspace{2pt}s\right)$ & \hspace{50pt}3.13 KJ
 \\
 \hline
EAB$\left(0.7,8\hspace{2pt}s\right)$  & \hspace{50pt}1.95 KJ
\\
\hline
EAB$\left(0.9,4\hspace{2pt}s\right)$ & \hspace{50pt}0.93 KJ
\\
\hline
\end{tabular}
\label{power-for-rar-table}
\end{table}

\subsection{Trade off Curves and Corrections to 3GPP Settings}
The plots for $P_S^{max}$, $\overline E_{c}^{min}$ and $\overline T_{AD}^{min}$ obtainable from EAB-BF (as obtained by our analysis in Table \ref{table:4}) against the set of lower bounds on $P_{min}$ are shown in Fig. \ref{optimize-together}(a), Fig. \ref{optimize-together}(b) and Fig. \ref{optimize-together}(c) respectively. These figures also show the performances that the three 3GPP settings (see tables \ref{table:3} and \ref{power-for-rar-table} for the 3GPP performances) can provide. Note that the performances of 3GPP settings are shown as straight lines to indicate that their performance do not depend on the desired $P_{min}$. Figures \ref{optimize-together}(d), \ref{optimize-together}(e) and \ref{optimize-together}(f) show the gains in success probability, energy consumption and mean access delay, respectively, that the near-optimal settings can provide against these three 3GPP settings.
It can be seen from these figures that the energy consumption and the access delay can be reduced by upto 50\% with more than 50\% gain in success probability compared to EAB$\left(0.9,4\hspace{2pt}s\right)$. In fact, all the obtainable near-optimal 3-tuples always perform better than  EAB$\left(0.9,4\hspace{2pt}s\right)$ till $P_{min}$ touches about 0.7 indicated by Point A in figures \ref{optimize-together}(b), \ref{optimize-together}(c), \ref{optimize-together}(e) and \ref{optimize-together}(f). Beyond Point A, EAB$\left(0.9,4\hspace{2pt}s\right)$ starts performing better than any optimal 3-tuples and the other two 3GPP settings in terms of energy consumption and access delay, but its success probability is always the worst. A reduction of 50\%  in energy consumption and around 20\% in access delay against EAB$\left(0.7,8\hspace{2pt}s\right)$ can be obtained without degrading the success probability indicated by Point B in all plots in Fig. \ref{optimize-together}. Similarly, a reduction of 50\%  in 
energy consumption and more than 50\% in access delay against EAB$\left(0.5,16\hspace{2pt}s\right)$ can be obtained without degrading the success probability indicated by Point C in all plots in Fig. \ref{optimize-together}. Hence, the 3GPP settings can be corrected by using the near-optimal barring parameter pairs $\left(\widehat P_{eab}, \widehat T_{eab}\right)$ that provide these gains and can be obtained from Table \ref{table:4}.

\section{Conclusion}\label{conclusion-section}
3GPP proposed Extended Access Barring (EAB) as the baseline solution to mitigate the RAN overload due to synchronized Random Access Channel (RACH) attempts by IoT devices in LTE-A. It was suggested to announce the EAB information either through a barring factor (EAB-BF) or a barring bitmap (EAB-BB). EAB-BB was adopted by 3GPP. 

In this work, we developed a novel analytical model to obtain the performance metrics of EAB-BF. Our analysis results were validated through simulations. Furthermore, we also developed the eNodeB energy consumption model to serve the IoT RACH requests in a LTE-A cell. Our analytical expressions along with the energy consumption model of the eNodeB help to build search algorithms to obtain EAB-BF settings that can simultaneously minimize eNodeB energy consumption, maximize success probability and minimize mean access delay for IoT devices. From the results obtained through an example search algorithm, we then showed that the optimal performance of EAB-BF is better than that of EAB-BB. Furthermore, we also showed the three 3GPP-proposed settings that were considered during standardization for EAB-BF provide sub-optimal QoS to devices and also result in excessive eNodeB energy 
consumption. 
We then showed how the 3GPP settings can be corrected that could lead to significant gains in performance of barring factor enabled EAB.  
\appendices
\section{Proof of Claim 1}\label{proof1}
 We can write the following using the total expectation theorem \cite{david-williams} , 
 \begin{equation}
  \mathbb{E}[A_i] = \sum_{\forall n}^{} \mathbb{E}[A_i|\varGamma_i = n]P(\varGamma_i=n).
 \end{equation}
 \vspace{-0.1mm}
 Since the devices try to clear the EAB test independently, we have,
\begin{equation}
\label{elements-of-T}
 P(A_i=m|\varGamma_i=n) = {n \choose m} P_{eab}^m \left(1-P_{eab}\right)^{n-m}.
\end{equation}
 Hence, we have, $\mathbb{E}[A_i|\varGamma_i = n] = n P_{eab}$.
Therefore,
\begin{equation} \label{eqn-for-Ai}
\mathbb{E}[A_i] = \sum_{\forall n}^{} n P_{eab}P(\varGamma_i=n) = P_{eab}\mathbb{E}[\varGamma_i].\\
\end{equation}
To find $\mathbb{E}[\varGamma_i]$,  we note that the total number of arrivals in slot $i$ is given by:
\begin{equation}
\label{eqn-for-vargammai}
 \varGamma_i =  F_i + \sum_{l=i-\frac{W}{r_p}}^{i-1} B_{l,i} + \varGamma^{'}_{i-\frac{T_{eab}}{r_p}},\hspace{5pt} i>0; \hspace{5pt} F_i=0, i>M,
\end{equation}where, $\varGamma^{'}_{i-\frac{T_{eab}}{r_p}}$ represents the fraction of devices out of $\varGamma_{i-\frac{T_{eab}}{r_p}}$ in slot $i-\frac{T_{eab}}{r_p}$ that could not pass the EAB test, backed off for time $T_{eab}$ and returned in the current slot $i$. The term $B_{l,i}$ represents the fraction of devices out of $C_l$ that arrive in slot $i$ after colliding in slot $l$. Since a device that collides backs off randomly over a window of length $W$ and the RACH periodicity is $r_p$, only those devices whose preambles collided within the previous $\frac{W}{r_p}$ slots can possibly re-attempt in the current slot $i$.

We now try to express $\varGamma_i$ only in terms of $F_i$ and $B_{l,i}$. To do that, we use the notation $q \triangleq \frac{T_{eab}}{r_p}$ and let {\small$i \in \left[Qq+1,(Q+1)q\right]$; $Q \in \mathbb{Z}^+$}. Then, \eqref{eqn-for-vargammai} can be written as:
\begin{equation}\label{eqn-for-vargammai-short}
 \varGamma_i =  F_i + \sum_{l=i-\frac{W}{r_p}}^{i-1} B_{l,i} + \varGamma^{'}_{i-q}.
\end{equation}
Now, we can write an expression for $\varGamma_{i-q}$,  by replacing $i$ with $i-q$ in \eqref{eqn-for-vargammai-short} to get,
\begin{equation}\label{eqn-for-vargammai-q}
  \varGamma_{i-q} = F_{i-q} + \sum_{l=(i-q)-\frac{W}{r_p}}^{(i-q)-1} B_{l,i-q} + \varGamma^{'}_{i-2q}. 
\end{equation}
Hence, we can now plug in $\varGamma_{i-q}$ from \eqref{eqn-for-vargammai-q} into \eqref{eqn-for-vargammai-short} to get,
\begin{equation}
  \begin{split}
  \varGamma_i &= F_i + \sum_{l=i-\frac{W}{r_p}}^{i-1} B_{l,i} + \left( F_{i-q} + \sum_{l=(i-q)-\frac{W}{r_p}}^{(i-q)-1} B_{l,i-q} + \varGamma^{'}_{i-2q}\right)^{'}\\
  &= F_i + \sum_{l=i-\frac{W}{r_p}}^{i-1} B_{l,i} + \left( F_{i-q} + \sum_{l=(i-q)-\frac{W}{r_p}}^{(i-q)-1} B_{l,i-q}\right)^{'} + \varGamma^{''}_{i-2q}.
    \end{split}
\end{equation}

Repeating this procedure recursively to substitute for $\varGamma_{i-2q}, \varGamma_{i-3q},....,\varGamma_{i-Qq}$ and noting that $\varGamma_{i-(Q+1)q}=0$, we get,
\begin{equation}
  \begin{split}
  \varGamma_i &= \left(F_i + \sum_{l=i-\frac{W}{r_p}}^{i-1} B_{l,i}\right) + \left(F_{i-q} + \sum_{l=(i-q)-\frac{W}{r_p}}^{(i-q)-1} B_{l,i-q}\right)^{'}\\
  &+\left(F_{i-2q} + \sum_{l=(i-2q)-\frac{W}{r_p}}^{(i-2q)-1} B_{l,i-2q}\right)^{''}+........\\
  &+\left(F_{i-Qq} + \sum_{l=(i-Qq)-\frac{W}{r_p}}^{(i-Qq)-1} B_{l,i-Qq}\right)^{\overbrace{'''''}^{\text{Q times}}},
  \end{split}
\end{equation}where, $(X)^{\overbrace{'''''}^{\text{n times}}}$ denotes the fraction of devices out of $X$ that could not pass the EAB test in $n$ successive attempts. 

Since the devices try to pass the EAB test independently and also each attempt of a device itself is independent of its previous attempts, after some rearrangement of terms and denoting $(X)^{\overbrace{'''''}^{\text{n times}}}$ as $(X)^{n'}$, we can express $\varGamma_i$ as:
\vspace{-0.05in}
\begin{equation}
  \begin{split}
  \varGamma_i &= F_i + F_{i-q}^{1'} + F_{i-2q}^{2'} +... + F_{i-Qq}^{Q'} + \sum_{l=i-\frac{W}{r_p}}^{i-1} B_{l,i} \\
  & + \sum_{l=(i-q)-\frac{W}{r_p}}^{(i-q)-1} B_{l,i-q}^{1'} + \sum_{l=(i-2q)-\frac{W}{r_p}}^{(i-2q)-1} B_{l,i-2q}^{2'}\\
  &+...+ \sum_{l=(i-Qq)-\frac{W}{r_p}}^{(i-Qq)-1} B_{l,i-Qq}^{Q'}.  
  \end{split}
\end{equation}
which can then be succinctly expressed as:
\begin{equation}
  \varGamma_i = \sum_{j=0}^{Q} \left[ F_{i-jq}^{j'} + \sum_{l=i-\left(jq+\frac{W}{r_p}\right) }^{ i-\left(jq+1\right)}B_{l,i-jq}^{j'}\right].
\end{equation}

Hence,
\begin{equation}\label{eqn-for-exp-of-vargammi}
\begin{split}
 \mathbb{E}\left[\varGamma_i\right] &= \mathbb{E} \sum_{j=0}^{Q} \left[ F_{i-jq}^{j'} + \sum_{l=i-\left(jq+\frac{W}{r_p}\right) }^{ i-\left(jq+1\right)}B_{l,i-jq}^{j'}\right]\\
 & \overset{(a)}{=}\sum_{j=0}^{Q} \left[ \mathbb{E}\left[F_{i-jq}^{j'}\right] + \sum_{l=i-\left(jq+\frac{W}{r_p}\right) }^{ i-\left(jq+1\right)}\mathbb{E}\left[B_{l,i-jq}^{j'}\right]\right],
\end{split}
\end{equation}where, $(a)$ follows because of linearity of the expectation operator \cite{david-williams}. Now, we can write ,
\begin{equation}
\label{eqn-for-exp-of-f}
 \mathbb{E}\left[F_{i-jq}^{j'}\right] = \sum_{F_{i-jq}}^{} \mathbb{E}\left[F_{i-jq}^{j'}|F_{i-jq}\right]P(F_{i-jq}).
\end{equation}

Using the law of iterated expectation (also called \emph{Tower Property of Conditional Expectation}) \cite{david-williams}, we can now write,

\begin{equation}
\begin{split}
\label{iterated-law}
 \mathbb{E}\left[F_{i-jq}^{j'}|F_{i-jq}\right] &\overset{(a)}{=} \mathbb{E}\left(\mathbb{E}\left[F_{i-jq}^{j'}|F_{i-jq}^{(j-1)'}\right]|F_{i-jq}\right)\\
 &\overset{(b)}{=} \mathbb{E}\left(\left(1-P_{eab}\right)F_{i-jq}^{(j-1)'}|F_{i-jq}\right)\\
 &\overset{(c)}{=} \left(1-P_{eab}\right)\mathbb{E}\left( F_{i-jq}^{(j-1)'}|F_{i-jq}\right)\\
 &\overset{(d)}{=} \left(1-P_{eab}\right)^{j-1}\mathbb{E}\left( F_{i-jq}^{'}|F_{i-jq}\right)\\
 &=\left(1-P_{eab}\right)^{j}F_{i-jq}.
 \end{split}
\end{equation}
Here, $(a)$ follows because $F_{i-jq}^{j'}$ depends on $F_{i-jq}^{(j-1)'}$ and is binomially distributed. $(b)$ follows because a device fails to clear the EAB test with probability $1-P_{eab}$ and the devices attempt to clear the EAB test independently. $(c)$ follows because expectation is a linear operator and $(d)$ follows from repeated application of law of iterated expectation.
From \eqref{eqn-for-exp-of-f}, it then follows that,
\begin{equation}
\label{final-eqn-for-exp-of-fjbar}
 \mathbb{E}\left[F_{i-jq}^{j'}\right] = \left(1-P_{eab}\right)^{j}\mathbb{E}\left[F_{i-jq}\right].
\end{equation}
We can similarly write,
\begin{equation}
 \mathbb{E}\left[B_{l,i-jq}^{j'}\right] = \sum_{B_{l,i-jq}}^{} \mathbb{E}\left[B_{l,i-jq}^{j'}|B_{l,i-jq}\right]P(B_{l,i-jq}).
\end{equation}
Following the same arguments as in \eqref{iterated-law}, we can write,
\begin{equation}
\label{eqn-for-exp-of-clbarjdash}
 \mathbb{E}\left[B_{l,i-jq}^{j'}\right] = \left(1-P_{eab}\right)^{j}\mathbb{E}\left[B_{l,i-jq}\right].
\end{equation}
Now, we can write,
\begin{equation}
\label{eqn-for-exp-of-clbar}
 \mathbb{E}\left[B_{l,i-jq}\right] = \sum_{\forall C_l: (i-jq) - l \leq \frac{W}{r_p}}^{} \mathbb{E}\left[B_{l,i-jq}|C_{l}\right]P(C_{l}).
\end{equation}
Since backoff time of the devices that collide is unformly distributed over the interval $\left[0,W-1\right]$, we have,
\begin{equation}
 \mathbb{E}\left[B_{l,i-jq}|C_{l}\right] = \frac{r_p}{W} C_l, \hspace{3pt} (i-jq) - l \leq \frac{W}{r_p}.
\end{equation}
Plugging this result in \eqref{eqn-for-exp-of-clbar} gives,
\begin{equation}
\label{final-eqn-for-exp-of-clbar}
 \mathbb{E}\left[B_{l,i-jq}\right] = \frac{r_p}{W} \mathbb{E}\left[C_l\right].
\end{equation}

Plugging \eqref{final-eqn-for-exp-of-clbar} into \eqref{eqn-for-exp-of-clbarjdash}, we get,
\begin{equation}
\label{final-eqn-for-exp-of-clbarjdash}
 \mathbb{E}\left[B_{l,i-jq}^{j'}\right] = \left(1-P_{eab}\right)^{j}\frac{r_p}{W} \mathbb{E}\left[C_l\right].
\end{equation}
From \eqref{eqn-for-Ai},  \eqref{eqn-for-exp-of-vargammi}, \eqref{final-eqn-for-exp-of-fjbar} and \eqref{final-eqn-for-exp-of-clbarjdash}, the result follows.

\section{Proof of Claim 2}\label{proof2}
We define an indicator random variable $X_k$ to indicate a preamble's success. $X_k = 1$ if the preamble $k, 1\leq k\leq K$, is chosen by only one among $A_i$ devices that attempt, else its value is 0. Hence, if $S_i$ denotes the number of successful preambles, then we can write,
\begin{equation}
\label{eqn-for-successes}
 S_i = \sum_{k=1}^{K}X_k \Rightarrow \mathbb{E}\left[S_i\right] = \sum_{k=1}^{K}\mathbb{E}[X_k] = K\mathbb{E}[X_k]
\end{equation}
The last equality holds since the devices choose the preambles independently. Now, the probability that the preamble $k$ is chosen by only one device given $A_i$ attempt is given by,
\begin{equation}
 P\left(X_k=1|A_i\right) = {A_i\choose 1}\frac{1}{K}\left(1-\frac{1}{K}\right)^{A_i-1},
\end{equation}since, each device chooses from the set of preambles with uniform distribution. Hence, 
\begin{equation}
\mathbb{E}\left[X_k|A_i\right] = {A_i\choose 1}\frac{1}{K}\left(1-\frac{1}{K}\right)^{A_i-1}.
\end{equation}
Therefore, from \eqref{eqn-for-successes}, we have, 
\begin{equation}
\label{eqn-for-conditional-exp-successes}
 \mathbb{E}\left[S_i|A_i\right] = K\mathbb{E}\left[X_k|A_i\right] = {A_i\choose 1}\left(1-\frac{1}{K}\right)^{A_i-1}.
\end{equation}
Also,
\begin{equation}
 \mathbb{E}\left[C_i|A_i\right] = A_i - \mathbb{E}\left[S_i|A_i\right].
\end{equation}

By applying the total expectation theorem, we can write,
\begin{equation}\label{eqn-for-Ci}
 \begin{split}
 \mathbb{E}\left[C_i\right] & = \sum_{\forall A_i}^{} \mathbb{E}\left[C_i|A_i\right] P\left(A_i\right)\\
 &= \mathbb{E}\left[A_i\right] - \sum_{\forall A_i}^{} \mathbb{E}\left[S_i|A_i\right]P\left(A_i\right).
 \end{split}
\end{equation}

Since the second term in \eqref{eqn-for-Ci} does not yield a closed form expression, we make some key observations to find a close approximation for it. Recall that,
\begin{displaymath}
 P\left(A_i = m\right) = \sum_{n=0}^{N} P(A_i = m|\varGamma_i=n).P(\varGamma_i=n).
\end{displaymath}
Therefore, only when $P_{eab}$ is high (see \eqref{elements-of-T}) and $P(\varGamma_i = n)$ is high for large $n$ (which is possible only with large $P_{eab}$ and small $T_{eab}$), $P\left(A_i\right)$ will exist for large $A_i$.  However, for a large $A_i$, $\mathbb{E}\left[S_i|A_i\right] \rightarrow 0$ (see Fig. \ref{observations-for-Ci}(a)) making the system unstable (since RACH procedure in LTE-A can be modeled as multi-channel Slotted ALOHA) \cite{osti,bertsekas} and the product 
$\mathbb{E}\left[S_i|A_i\right]P\left(A_i\right)$ negligible.
\begin{figure}
\centering
\mbox{\subfigure[EAB$\left(0.9,4s\right)$]{\includegraphics[width=1.75in]{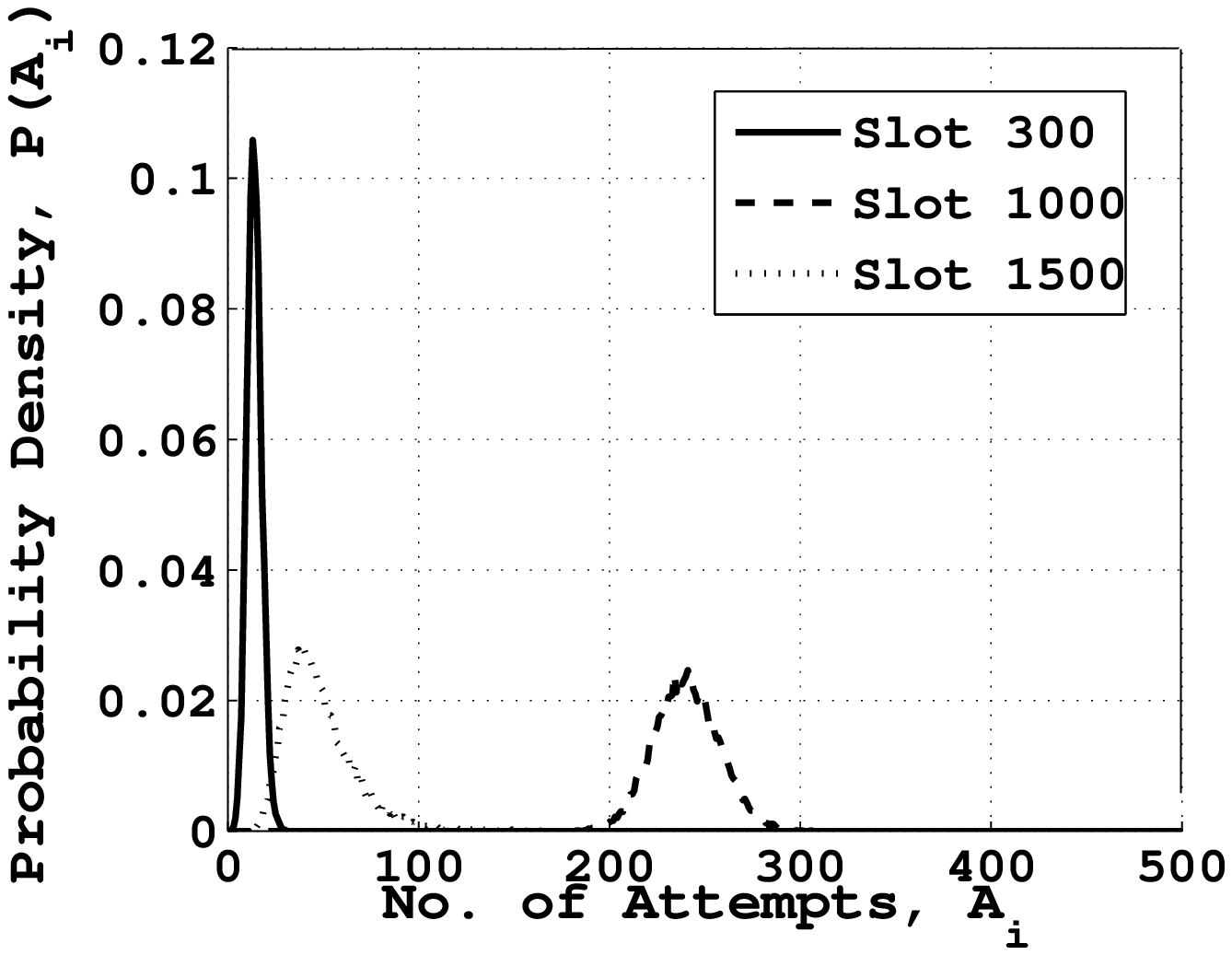}}\quad
\subfigure[EAB$\left(0.9,2s\right)$]{\includegraphics[width=1.75in]{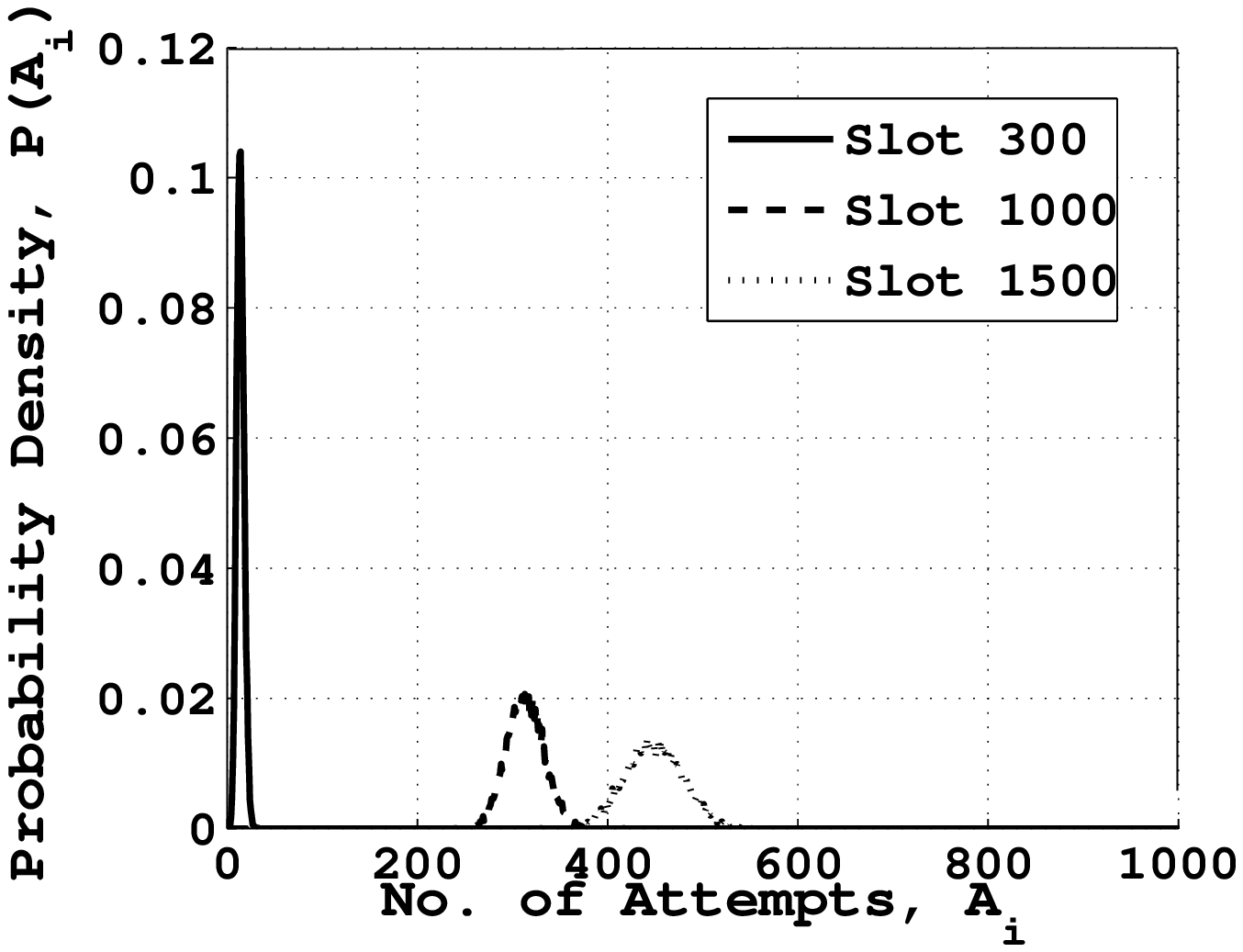} }}
\vspace*{-5mm}
\caption{Probability Distribution for $A_i$ for large $P_{eab}$ and small $T_{eab}$} 
\label{attempts-dist}
\end{figure}
Even though we do not have an exact form for $P\left(A_i\right)$, we run Monte Carlo simulations (25000 iterations) for settings shown in Table \ref{table:1} and plot $P\left(A_i\right)$ along Y-axis with $A_i$ along X-axis for EAB$\left(0.9,4\hspace{2pt}s\right)$ and EAB$\left(0.9,2\hspace{2pt}s\right)$ as shown in Fig. \ref{attempts-dist}. We plot $P\left(A_i\right)$ for slots 300,1000 and 1500 so that we can make observations at the beginning, during and after the activation interval of the devices. Clearly, even for large $P_{eab}$ and small $T_{eab}$ as shown in Fig. \ref{attempts-dist}(a) and \ref{attempts-dist}(b), $P\left(A_i\right)$ has negligible values beyond $A_i=100$. Additionally, $\mathbb{E}\left[S_i|A_i\right]$ falls off to small values beyond $A_i=100$ (see Fig. \ref{observations-for-Ci}(a)) to make their product negligible. It would then be reasonable to assume that for any combinations of smaller $P_{eab}$ and/or larger $T_{eab}$, the product of $P\left(A_i\right)$ and $\mathbb{E}\left[S_
i|A_i\right]$ would be 
negligible beyond $A_i=100$.
\begin{figure}
\centering
\mbox{\subfigure[Expected Number of Successes conditioned on $A_i$.]{\includegraphics[width=1.75in]{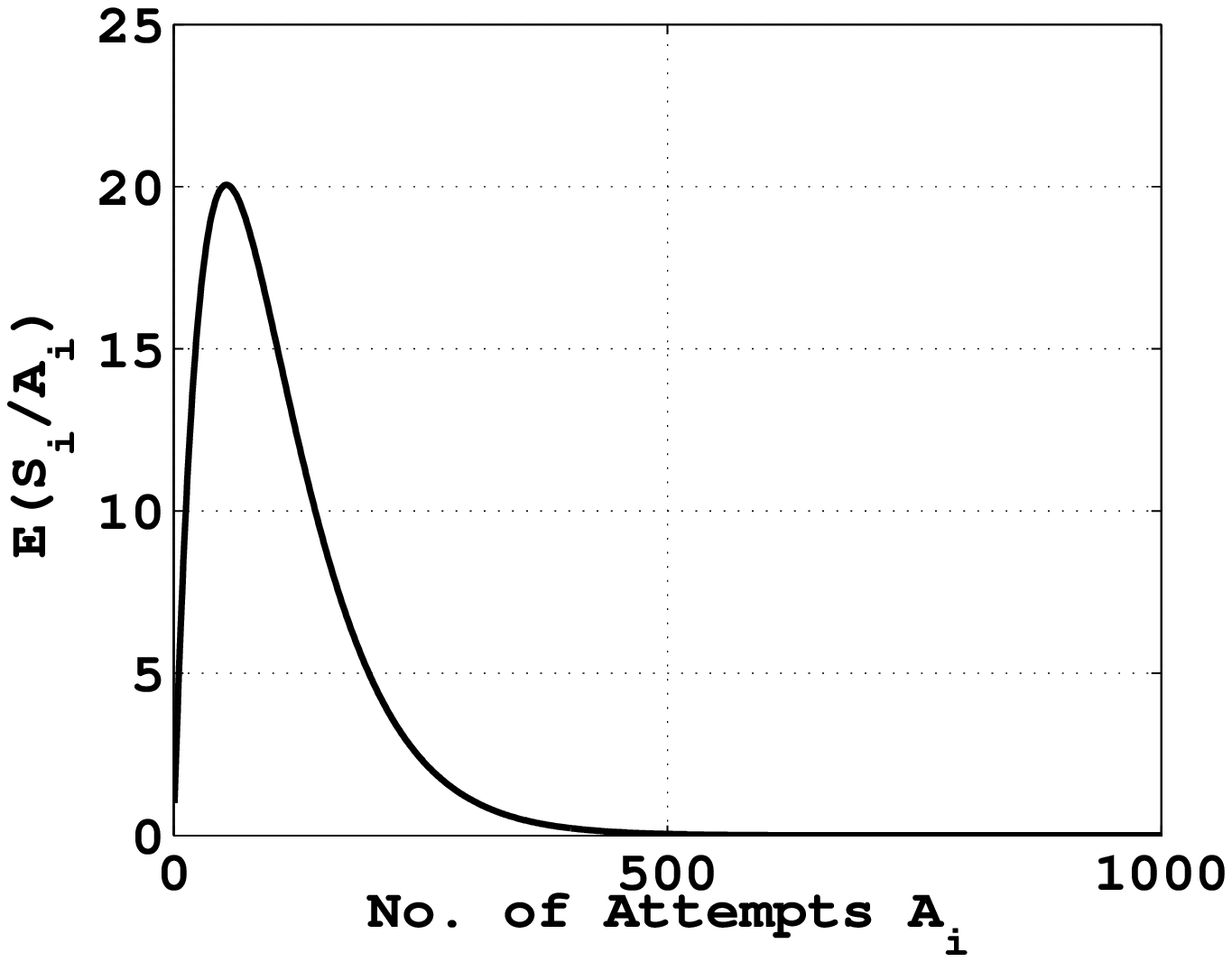}}\quad
\subfigure[Expectation of $C_i$ lies between the two curves shown.]{\includegraphics[width=1.75in]{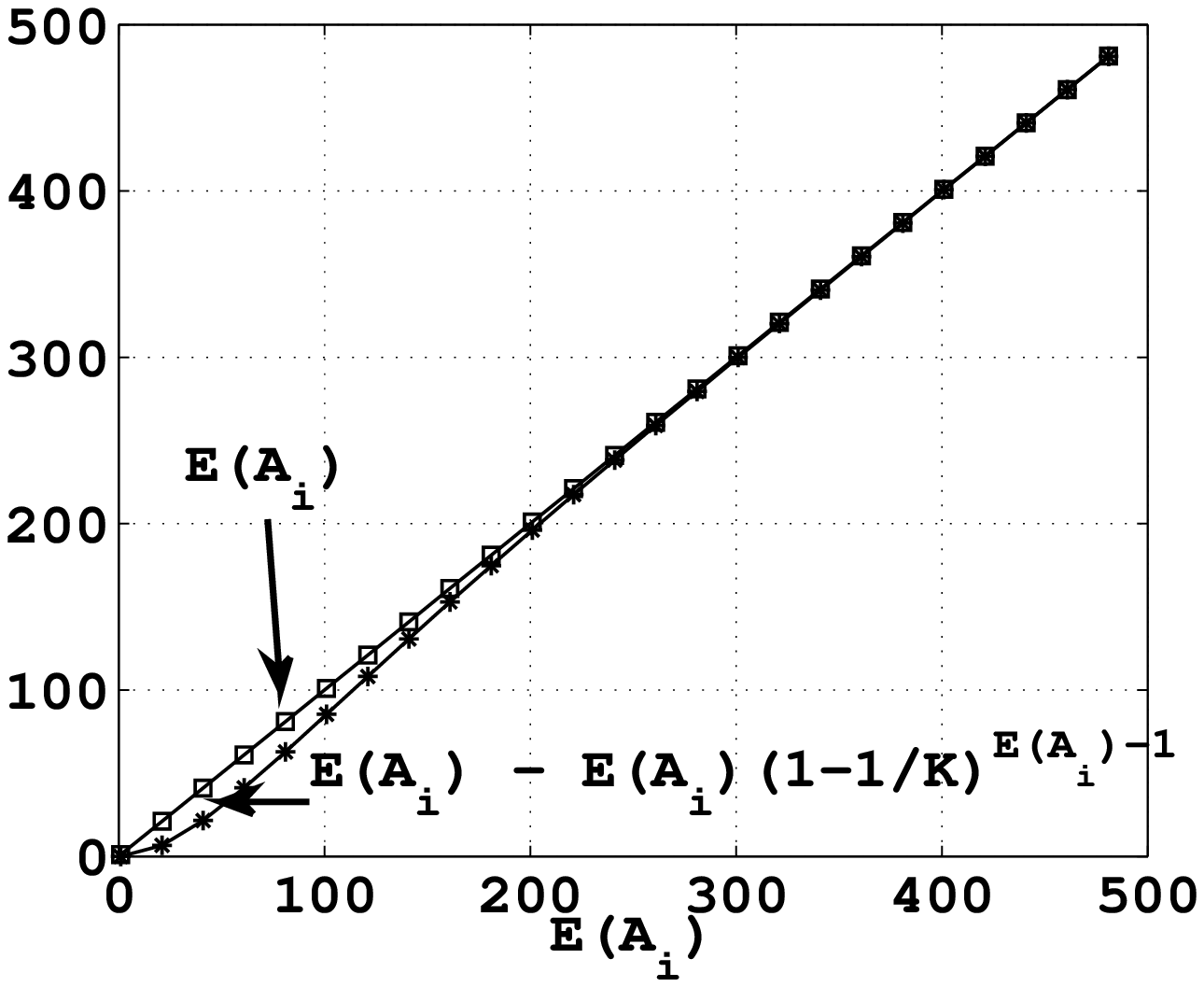} }}
\caption{Observations made for approximating $\mathbb{E}\left[C_i\right]$} 
\label{observations-for-Ci}
\end{figure}
We also observe that $\mathbb{E}\left[S_i|A_i\right]$ is quasi-concave since it is a function of only one variable and has only one peak as seen from its plot in Fig. \ref{observations-for-Ci}(a). Since $\mathbb{E}\left[S_i|A_i\right]$ is quasi-concave, we identify the range of $\left[0,A_{i}^{*}\right]$ over which the function is concave downward. The point $A_{i}^{*}$ can be found by equating the second differential of $\mathbb{E}\left[S_i|A_i\right]$ to zero. We identify that $A_{i}^{*}$ is close to 100 which is the value of $A_i$ beyond which we assumed the product  $P\left(A_i\right)\mathbb{E}\left[S_i|A_i\right]$ would be negligible. Hence, we can now write,
\begin{equation}
\small
 \begin{split}
 \sum_{\forall A_i}^{} \mathbb{E}\left[S_i|A_i\right]P\left(A_i\right) &=  \sum_{A_i \leq A_{i}{*}}^{} \mathbb{E}\left[S_i|A_i\right]P\left(A_i\right)+ \sum_{A_i > A_{i}^{*}}^{} \mathbb{E}\left[S_i|A_i\right]P\left(A_i\right)\\
 &\overset{(a)}{=}\sum_{A_i \leq A_{i}{*}}^{} \mathbb{E}\left[S_i|A_i\right]P\left(A_i\right)\\
 &\overset{(b)}{\leq} \mathbb{E}\left[A_i\right]\left(1-\frac{1}{K}\right)^{\mathbb{E}\left[A_i\right]-1},
 \end{split}
\end{equation}where, $(a)$ follows because of the above observations and $(b)$ follows because of \eqref{eqn-for-conditional-exp-successes} and Jensen's Inequality \cite{david-williams} since $\mathbb{E}\left[S_i|A_i\right]$ is concave over the range $\left[0,A_{i}^{*}\right]$.
Hence, from \eqref{eqn-for-Ci}, we then have,
\begin{equation}\label{final-bound-for-Ci}
  \mathbb{E}\left[C_i\right] \geq  \mathbb{E}\left[A_i\right] -  \mathbb{E}\left[A_i\right]\left(1-\frac{1}{K}\right)^{ \mathbb{E}\left[A_i\right]-1}.
\end{equation}

Even though the RHS of \eqref{final-bound-for-Ci} represents the lower bound for $ \mathbb{E}\left[C_i\right]$, clearly, $ \mathbb{E}\left[C_i\right]\leq  \mathbb{E}\left[A_i\right]$, since the expected number of collisions cannot exceed the expected number of attempts in a RACH slot. Therefore, 
\begin{equation}
\label{sandwich-Ci}
  \mathbb{E}\left[A_i\right] -  \mathbb{E}\left[A_i\right]\left(1-\frac{1}{K}\right)^{ \mathbb{E}\left[A_i\right]-1} \leq  \mathbb{E}\left[C_i\right] \leq  \mathbb{E}\left[A_i\right].
\end{equation}
Hence $\mathbb{E}\left[C_i\right]$ is sandwiched between the two curves as shown in Fig. \ref{observations-for-Ci}(b). Observe that the two curves are close to each other. Additionally, $\mathbb{E}\left[C_i\right]$ will be close to the lower bound when $\mathbb{E}\left[A_i\right]$ is low, and when $\mathbb{E}\left[A_i\right]$ is high $\mathbb{E}\left[A_i\right] -  \mathbb{E}\left[A_i\right]\left(1-\frac{1}{K}\right)^{ \mathbb{E}\left[A_i\right]-1} \rightarrow \mathbb{E}\left[A_i\right]$. Hence,  we can approximate $ \mathbb{E}\left[C_i\right]$ with the lower bound in \eqref{sandwich-Ci} to make the analysis tractable. Therefore, 
\begin{equation}
  \mathbb{E}\left[C_i\right] =  \mathbb{E}\left[A_i\right] -  \mathbb{E}\left[A_i\right]\left(1-\frac{1}{K}\right)^{ \mathbb{E}\left[A_i\right]-1}.
\end{equation}

%

%
%
%





\begin{thebibliography}{1}

\bibitem{Luigi}
 Luigi Atzori, Antonio Iera and Giacomo Morabito, ``The Internet of Things: A survey'', in \emph{Computer Networks}, (2010), Volume 54, Issue 15, Pages 2787-2805, October 2010, .

 
\bibitem{3G6}
3GPP TR 37.868 V11.0.0, ``Study on RAN Improvements for Machine-type Communications'', Technical Report, Sept 2011. 
 

\bibitem{cheng}
Jen-Po Cheng, Chia-han Lee and Tzu-Ming Lin, ``Prioritized Random Access with Dynamic Access Barring for RAN Overload in 3GPP LTE-A Networks'', in \emph{GLOBECOM Workshops (GC Wkshps)}, Dec. 2011, pp. 368–372.

\bibitem{cheng2}
Tzu-Ming Lin, Chia-Han Lee, Jen-Po Cheng and Wen-Tsuen Chen, ``PRADA: Prioritized Random Access With Dynamic Access Barring for MTC in 3GPP LTE-A Networks '', in \emph{IEEE Transactions on Vehicular Technology}, Volume:63 , Issue: 5, Pages: 2467 - 2472, Jan 2014. 

\bibitem{timer}
3GPP TS 36.413 version 8.7.0 Release 8, ``S1 Application Protocol (S1AP)'', Technical Specification, V8.7.0 (2009-10).

\bibitem{wong}
Suyang Duan Vahid Shah-Mansouri and Vincent W.S. Wong ``Dynamic Access Class Barring for M2M Communications in LTE Networks
'', in \emph{IEEE Global Communications Conference (GLOBECOM), 2013}, Pages 4747-4752, 9-13 Dec 2013, Atlanta, GA.


\bibitem{Feng}
Z. Feng and X. F. Zhong, ``Fast Retrial and Dynamic Access Control Algorithm for LTE-Advanced Based MTC Network'', in \emph{AICT, 2012}, ISBN: 978-1-61208-199-1, Pages 24-28.


\bibitem{Ki-Dong}
Ki-Dong Lee, Sang Kim, and Byung Yi, ``Throughput Comparison of Random Access Methods for MTC Service over LTE Networks'', \emph{IEEE Globecome Workshops}, E-ISBN :978-1-4673-0038-4, Pages 373-377, 5-9 Dec. 2011, Houston, TX. 


\bibitem{cheng3}
Ray-Guang Cheng, Jenhui Chen, Dan-Wu Chen, and Chia-Hung Wei, ``Modeling and Analysis of an Extended Access Barring Algorithm for Machine-Type Communications in LTE-A Network '', in \emph{IEEE Transactions on Wireless Communications}, Volume:14 , No:6, June 2015. 


\bibitem{zhang}
Zhang Zhang, Hua Chao, Wei Wang and Xun Li, ``Performance Analysis and UE-side Improvement of Extended Access Barring for Machine Type Communications in LTE '', in \emph{IEEE 79th Vehicular Technology Conference (VTC Spring), 2014}, Pages 1-5, 18-21 May 2014, Seoul. 



\bibitem{survey}
Ziaul Hasan, Hamidreza Boostanimehr, and Vijay K.Bhargava, ``Green Cellular Networks, A Survey, Some Research Issues and Challenges'', in  \emph{IEEE Communications Surveys \& Tutorials}, vol.\ 13, Issue 4, pp.~524--540, Nov. 2011. 
 

\bibitem{frenger}
Frenger, P, et.al, ``Reducing Energy Consumption in LTE with Cell DTX'', in \emph{IEEE 73rd Vehicular Technology Conference (VTC Spring), 2011 },  pp. 1-5, 15-18 May 2011, Yokohama. 
 
 
\bibitem{strinati}
Emilio Calvanese Strinati and Paolo Greco, ``Green Resource Allocation for OFDMA Wireless'', \emph{IEEE 21st International Symposium on Personal Indoor and Mobile Radio Communications (PIMRC), 2010},  pp. 2775 - 2780 , 26-30 Sept. 2010, Instanbul.


\bibitem{rohit}
Rohit Gupta and Emilio Calvanese Strinati, ``Base-Station Duty-Cycling and traffic buffering as a means to achieve Green Communications ,'' \emph{Vehicular Technology Conference (VTC Fall), 2012 IEEE},  pp. 1 - 6, 3-6 Sept. 2012, Quebec City, QC.

\bibitem{ps-sps}
Prashant Wali and Debabrata Das, ``PS-SPS: Power Saving-Semi Persistent Scheduler for VoLTE in LTE-Advanced'', \emph{IEEE International Conference on Electronics, Computing and Communication Technologies (CONECCT), 2015}, 10 Jul - 11 Jul 2015, Bangalore.

\bibitem{Gerasimenko}
Mikhail Gerasimenko, et.al, ``Energy and Delay Analysis of LTE-Advanced RACH Performance under MTC Overload'', in \emph{IEEE Global Workshops, 2012}, Pages 1632-1637, 3-7 Dec 2012, Anaheim, CA.


\bibitem{bhaskar}
Eunsung Oh and Bhaskar Krishnamachari, ``Energy Savings through Dynamic Base Station Switching in Cellular Wireless Access Networks'', in \emph{2010 IEEE  Global Telecommunications Conference (GLOBECOM 2010)}, pp. 1 - 5, 6-10 Dec. 2010, Miami.



\bibitem{3gpp-acb}
3GPP TS 36.331 version 12.3.0 Release 12, ``LTE; Evolved Universal Terrestrial Radio Access (E-UTRA); Radio Resource Control (RRC); Protocol specification'', Technical Specification, V8.7.0, 2014.

















\bibitem{stefania}
Stefania Sesia, Issam Toufik and Matthew Baker, \emph{LTE – The UMTS Long Term Evolution From Theory to Practice},  2009 John Wiley \& Sons, Ltd. ISBN: 978-0-470-69716-0.


\bibitem{TS-331}
3GPP TS 36.331 version 10.7.0 Release 10, ``LTE; Evolved Universal Terrestrial Radio Access (E-UTRA); Radio Resource Control (RRC); Protocol specification'', Technical Specification. 

\bibitem{3G9}
ETSI TS 136 213 V8.8.0 (2009-10), `LTE Evolved Universal Terrestrial Radio Access (E-UTRA) Physical layer procedures '', 3GPP TS 36.213 version 8.8.0 Release 8, Technical Specification, (2009-2010). 

\bibitem{3G10}
ETSI TS 136 212 V8.8.0 , `LTE Evolved Universal Terrestrial Radio Access (E-UTRA) Multiplexing and channel coding '', 3GPP TS 36.212 version 8.8.0 Release 8, Technical Specification, 2010-01).

\bibitem{3G8}
3GPP TS 36.321 V8.0.0 , `Evolved Universal Terrestrial Radio Access (E-UTRA)Medium Access Control (MAC) protocol specification'', Technical Specification, (2007-2012). 


\bibitem{minlp}
P. Belotti, C. Kirches, S. Leyffer, J. Linderoth, J. Luedtke and A. Mahajan, ``Mixed-Integer Nonlinear Optimization'', in  \emph{Acta Numerica 22}, pages 1-131, 2013.

\bibitem{boyd}
Stephen Boyd and Lieven Vandenberghe, \emph{Convex Optimization}, Edition 1, Cambridge University Press, 2004.

\bibitem{david-williams}
David Williams, \emph{Probability with Martingales}, Cambridge University Press (14 February 1991).

\bibitem{osti}
Prajwal Osti, Pasi Lassila, Samuli Aalto, Anna Larmo and Tuomas Tirronen, ``Analysis of PDCCH performance for M2M traffic in LTE'', in \emph{IEEE Transactions on Vehicular Technology}, Volume:PP, Issue:99, April 2014.


\bibitem{bertsekas}
D. Bertsekas and R. Gallager, \emph{Data Networks}, 2nd Edition, Prentice-Hall, 1992.




\end{thebibliography}
\end{document}